\documentclass[letter,envcountsect]{llncs}

\usepackage{amsmath,amssymb,bm, amsfonts, amsbsy, pstricks,bbm,mathrsfs, multirow, bigstrut,graphicx,fixmath, enumerate}
\usepackage[english]{babel} 
\usepackage{proof} 

\newcommand{\labelAlphabet}{\Gamma}
\newcommand{\zigzagnumber}{{zn}}
\newcommand{\sizeSliceGraph}{{r}}
\newcommand{\widthG}{{q}}
\newcommand{\interpretedAlphabet}{\slicealphabet^{c,\mathcal{X}}}
\newcommand{\interpretedAlphabetMinusX}{\slicealphabet^{c,\mathcal{X}\backslash \{X\}}}
\newcommand{\msotwo}{{$\mbox{MSO}_2$\;}}
\newcommand{\msoone}{{$\mbox{MSO}_1$\;}}

\newcommand{\boldS}{{\mathbf{S}}}
\newcommand{\boldU}{{\mathbf{U}}}
\newcommand{\unitdecomposition}{{\mathbf{U}}}
\newcommand{\lang}{{\mathcal{L}}} 
\newcommand{\slicegraphvertex}{\mathfrak{v}}

\newcommand{\N}{{\mathbold{N}}}

\newcommand{\slicegraph}{\mathcal{S}\!\mathcal{G}}
\newcommand{\slicealphabet}{\Sigma_{\mathbb{S}}}

\newcommand{\keywords}[1]{\par\addvspace\baselineskip\noindent\keywordname\enspace\ignorespaces#1}

\newcommand{\graph}{{\mathcal{G}}}

\newcommand{\subslicesNumber}{{\mathcal{S}\mathcal{U}\mathcal{B}\mathcal{N}}}

\newcommand{\emptyslice}{{\bm{\varepsilon}}}

\newcommand{\subgraphsSlicegraph}{\mathcal{S}\mathcal{U}\mathcal{B}}

\newcommand{\numberingExtension}{{\mathcal{N}}}

\begin{document}

% \title[short title]{title}
%\title[Using stacs.cls]{How to use stacs.cls}
\title{Subgraphs Satisfying MSO Properties on $z$-Topologically Orderable Digraphs}
% \author[ref]{Short author}{Author}
\author{Mateus de Oliveira Oliveira}
\institute{School of Computer Science and Communication, \\ KTH Royal Institute of Technology, 
100-44 Stockholm, Sweden\\ mdeoliv@kth.se}

\maketitle

\vspace{-20pt}
\begin{abstract} \noindent
We introduce the notion of $z$-topological orderings for digraphs. 
We prove that given a digraph $G$ on $n$ vertices admitting a $z$-topological ordering, 
together with such an ordering, one may count the number of subgraphs of $G$ that at the same time satisfy a monadic second order 
formula $\varphi$ and are the union of $k$ {\bf directed} paths, 
in time $f(\varphi,k,z)\cdot n^{O(k\cdot z)}$. 
Our result implies the polynomial time solvability of many  
natural counting problems on digraphs admitting $z$-topological orderings 
for constant values of $z$ and $k$. 
Concerning the relationship between 
$z$-topological orderability and other digraph width measures, we observe that 
any digraph of {\bf directed} path-width $d$ has a $z$-topological ordering 
for $z\leq 2d+1$. On the other hand, there are digraphs on $n$ vertices admitting a 
$z$-topological order for $z=2$, but whose directed path-width is $\Theta(\log n)$. Since graphs of bounded {\bf directed} path-width can have both 
arbitrarily large {\bf undirected} tree-width and arbitrarily large clique width, 
our result provides for the first time a suitable way of partially transposing metatheorems 
developed in the context of the monadic second order logic of graphs of 
constant {\bf undirected} tree-width and constant clique width to the realm of digraph width measures
that are closed under taking subgraphs and whose constant levels 
incorporate families of graphs of arbitrarily large undirected tree-width and 
arbitrarily large clique width. 
\vspace{-10pt}
\end{abstract}
\vspace{-20pt}

\section{Introduction} 
\label{section:Introduction} 
\vspace{-10pt}
Two cornerstones of parametrized complexity theory are Courcelle's theorem \cite{Courcelle1990}
stating that monadic second order logic properties may be model checked in linear time in graphs of constant undirected tree-width, and its 
subsequent generalization to counting given by Arnborg, Lagergren and Seese \cite{ArnborgLagergrenSeese1991}.
The importance of such metatheorems stem from the fact that several NP-complete problems 
such as Hamiltonicity, colorability, and their respective $\#$P-hard counting counterparts, can be modeled in terms 
of \msotwo sentences and thus can be efficiently solved in graphs of constant {\bf undirected} tree-width. 

In this work we introduce the notion of $z$-topological orderings for digraphs and provide a suitable way of partially transposing the metatheorems 
in \cite{Courcelle1990,ArnborgLagergrenSeese1991} to digraphs admitting {\em $z$-topological orderings} for constant values of $z$. 
In order to state our main result we will first give a couple of easy definitions: 
Let $G=(V,E)$ be a directed graph. For subsets of vertices $V_1,V_2 \subseteq V$ we let $E(V_1,V_2)$ 
denote the set of edges with one endpoint in $V_1$ and another endpoint in $V_2$. We say that a linear ordering
$\omega= (v_1,v_2,...,v_n)$ of the vertices of $V$ is a $z$-topological ordering 
of $G$ if for every {\bf directed} simple path $p=(V_p,E_p)$ in $G$ and every $i$ with 
$1\leq i \leq n$, we have that $|E_p\cap E(\{v_1...,v_i\},\{v_{i+1},...,v_{n}\})| \leq z$.
In other words, $\omega$ is a $z$-topological ordering if every {\bf directed} simple path of $G$ bounces
back and forth at most $z$ times along $\omega$. The terminology $z$-topological ordering
is justified by the fact that any topological ordering of a DAG $G$ according to the usual 
definition, is a $1$-topological ordering according to our definition. Conversely 
if a digraph admits a $1$-topological ordering, then it is a DAG. 
We denote by $MSO_2$ the monadic second order logic of graphs with edge set quantification. An edge-weighting 
function for a digraph $G=(V,E)$ is a function $w:E\rightarrow \Omega$ where $\Omega$ is a finite commutative semigroup 
of size polynomial in $|V|$ whose elements are totally ordered. The weight of a subgraph $H=(V,E')$ of 
$G$ is defined as $w(H)=\sum_{e\in E'} w(e)$. A maximal-weight subgraph of $G$ satisfying a given
property $\varphi$ is a subgraph $H=(V,E')$ such that $H\models \varphi$ and such that for any other 
subgraph $H'=(V'',E'')$ of $G$ such that $H'\models \varphi$ we have $w(H)\geq w(H')$. 
Now we are in a position to state our main theorem:

\begin{theorem}[Main Theorem]
\label{theorem:MainTheorem}
For each \msotwo formula $\varphi$ and each positive integers $k,z\in \N$ there exists 
a computable function $f(\varphi,z,k)$ such that: Given a digraph $G=(V,E)$ of zig-zag number $z$ 
on $n$ vertices, a weighting function $w:E\rightarrow \Omega$, a $z$-topological ordering $\omega$ of $G$ and 
a number $l=O(n)$\footnote{Observe that $l$ can be as large as $n$.}, we can count in time $f(\varphi,z,k)\cdot n^{O(z\cdot k)}$ the number of 
subgraphs $H$ of $G$ simultaneously satisfying the following four properties: 
\begin{enumerate}[(i)]
	\item  $H\models \varphi$ \hspace{2.2cm} 
 	\item  $H$ is the union of $k$ directed paths\footnote{A digraph $H$ is the union of $k$ directed paths if $H=\cup_{i=1}^k p_i$ for not 
necessarily vertex-disjoint nor edge-disjoint directed paths $p_1,...,p_k$.}
	\item  $H$ has $l$ vertices\hspace{0.9cm}
	\item  $H$ has maximal weight
\end{enumerate}
%\begin{enumerate}[]
%	\item $1)$ $H\models \varphi$ \hspace{2.2cm} 
% 	$2)$ $H$ is the union of $k$ directed paths\footnote{A digraph $H$ is the union of $k$ directed paths if $H=\cup_{i=1}^k p_i$ for not 
%necessarily vertex-disjoint nor edge-disjoint directed paths $p_1,...,p_k$.}
%	\item $3$) $H$ has $l$ vertices\hspace{0.9cm}  $4$) $H$ has maximal (minimal) weight
%\end{enumerate}
%
\end{theorem}

Our result implies the polynomial time solvability of many natural counting problems on digraphs 
admitting $z$-topological orderings for constant values of $z$ and $k$. 
We observe that graphs admitting $z$-topological 
orderings for constant values of $z$ can already have simultaneously unbounded tree-width and 
unbounded clique-width, and therefore the problems that we deal with here cannot be tackled by the 
approaches in \cite{Courcelle1990,ArnborgLagergrenSeese1991,CourcelleMakowskyRotics2000}. 
For instance any DAG is $1$-topologically orderable. In particular, the $n\times n$ directed 
grid in which all horizontal edges are directed to the left and all vertical edges oriented 
up is $1$-topologically orderable, while it has both undirected tree-width $\Omega(n)$ and 
clique-width $\Omega(n)$. 

\vspace{-7pt}
\section{Applications}
\label{secttion:Application}
\vspace{-7pt}
To illustrate the applicability of Theorem \ref{theorem:MainTheorem} with a simple example, 
suppose we wish to count the number of Hamiltonian cycles on $G$. Then our formula $\varphi$ will express 
that the graphs we are aiming to count are cycles, namely, connected graphs in which 
each vertex has degree precisely two. Such a formula can be easily specified in \msotwo. 
Since any cycle is the union of two directed paths, we have $k=2$. Since we want all vertices to be 
visited our $l=n$. Finally, the weights in this case are not relevant, so  it is enough to set 
the semigroup $\Omega$ to be the one element semigroup $\{1\}$, and the weights of all edges to be $1$.
In particular the total weight of any subgraph of $G$ according to this semigroup will be $1$. 
By Theorem \ref{theorem:MainTheorem} we can count the number of Hamiltonian cycles in 
$G$ in time $f(\varphi,k,z)\cdot n^{2 z}$. We observe that Hamiltonicity can be 
solved within the same time bounds for other directed width measures, such as directed tree-width \cite{JohnsonRobertsonSeymourThomas2001}.

Interestingly, Theorem \ref{theorem:MainTheorem} allow us to count structures that are much more 
complex than cycles. And in our opinion it is rather surprising that counting such complex structures 
can be done in XP. For instance, we could choose to count the number of maximal Hamiltonian subgraphs
of $G$ which can be written as the union of $k$ directed paths. We can repeat this trick with virtually 
any natural property that is expressible in \msotwo. For instance we can count the number of 
maximal weight $3$-colorable subgraphs of $G$ that are the union of $k$-paths. Or the number of subgraphs
of $G$ that are the union of $k$ directed paths and have di-cuts of size $k/10$. Observe that our framework 
does not allow one to find a maximal di-cut of the whole graph $G$ nor to determine in polynomial time 
whether the whole graph $G$ is 3-colorable, since these problems are already NP-complete for DAGs, i.e.,
for $z=1$. 

If $H=(V,E)$ is a digraph, then the disorientation of $H$ is the undirected graph $H'$ obtained from $H$ by 
forgetting the orientation of its edges. In other words, we add an edge $(v,v')$ to $E$ whenever $(v',v)\in E$.
A very interesting application of Theorem \ref{theorem:MainTheorem} consists in counting the number
of maximal-weight subgraphs of $G$ which are the union of $k$ paths and whose disorientation satisfy some structural
property, such as, connectedness, planarity, bounded genus, bipartiteness, etc. The proof of the next corollary can 
be found in Appendix \ref{section:ProofsLowSections}.  

\begin{corollary}
\label{corollary:StructuralCounting}
Let $G=(V,E)$ be a digraph on $n$ vertices and $w:E\rightarrow\Omega$ be an edge weighting function. Then 
given a $z$-topological ordering $\omega$ of $G$ one may count in time $O(n^{k\cdot z})$ the number of maximal-weight subgraphs that 
are the union of $k$ directed paths and satisfy any combination of the following properties: 
$1)$ Connectedness, $2)$ Being a forest, $3)$ Bipartiteness, $4)$ Planarity, $5)$ Constant Genus $g$, $6)$ Outerplanarity, 
$7)$ Being Series Parallel, $8)$ Having Constant Treewidth $t$ $9$) Having 
Constant Branchwidth $b$, $10$) Satisfy any minor closed property.
\end{corollary}

The families of problems described above already incorporate a large number of natural combinatorial problems. 
However the monadic second order formulas expressing the problems above are relatively simple and can be 
written with at most two quantifier alternations. As Matz and Thomas have shown however, the monadic second order 
alternation hierarchy is infinite \cite{MatzThomas1997}. Additionally Ajtai, Fagin and Stockmeyer 
showed that each level $r$ of the polynomial hierarchy has a very natural complete problem, 
the $r$-round-$3$-coloring problem, that belongs to the $r$-th level of the monadic second order hierarchy
(Theorem 11.4 of \cite{AjtaiFaginStockmeyer2000}). Thus by Theorem \ref{theorem:MainTheorem} 
we may count the number of $r$-round-3-colorable subgraphs of $G$ that are the union of $k$ directed paths 
in time $f(\varphi_{r},z,k)\cdot n^{O(z\cdot k)}$.

We observe that the condition that the subgraphs we consider are the union of $k$ directed paths is 
not as restrictive as it might appear at a first glance. For instance one can show that for any $a,b\in N$
the $a\times b$ undirected grid is the union of $4$ directed paths. Additionally these grids have 
zig-zag number number $O(\min\{a,b\})$. Therefore counting the number of maximal grids of height $O(z)$
on a digraph of zig-zag number $z$ is a neat example of problem which can be tackled 
by our techniques but which cannot be formulated as a linkage problem, namely, the most successful 
class of problems that has been shown to be solvable in polynomial time for constant values of 
several digraph width measures \cite{JohnsonRobertsonSeymourThomas2001}. 
\vspace{-10pt}

\section{Overview of the Proof of Theorem \ref{theorem:MainTheorem}}
\label{section:Overview}
\vspace{-5pt}
We will prove Theorem \ref{theorem:MainTheorem} within the framework of regular slice languages, 
which was originally developed by the author to tackle several problems
within the partial order theory of concurrency \cite{deOliveiraOliveira2010,deOliveiraOliveira2012}.
The main steps of the proof of Theorem \ref{theorem:MainTheorem} are as follows. To each regular slice 
language $\lang$ we associate a possibly infinite set of digraphs $\lang_{\graph}$. 
In Section \ref{section:MainTheorem} we will define the notion of $z$-dilated-saturated regular slice language
and show that given any digraph $G$ together with a $z$-topological ordering 
$\omega=(v_1,v_2,...,v_n)$ of $G$, and any $z$-dilated-saturated slice language $\lang$, one may efficiently count 
the number of subgraphs of $G$ that are isomorphic to some digraph in $\lang_{\graph}$ ( Theorem \ref{theorem:MainTechnical}). 
Then in Section \ref{section:MSO} we will show that given any monadic second order formula $\varphi$
and any natural numbers $z,k$ one can construct a $z$-dilated-saturated regular slice language $\lang(\varphi,z,k)$ representing
the set of all digraphs that at the same time satisfy $\varphi$ and are the union of $k$ directed paths (Theorem \ref{theorem:MonadicTechnical}). 
The construction of $\lang(\varphi,z,k)$ is done once and for all for each $\varphi$,$k$ and $z$, and is completely independent
from the digraph $G$. Finally, the proof of 
Theorem \ref{theorem:MainTheorem} will follow by plugging Theorem \ref{theorem:MonadicTechnical} into Theorem \ref{theorem:MainTechnical}. 
Proofs of intermediate results omitted for a matter of clarity or due to lack of space can be found in the appendix. 
%In the appendix we also state and prove a weighted version of Theorem \ref{theorem:MainTheorem} that can be 
%applied to maximization/minimization problems (Theorem \ref{theorem:MainTheoremWeighted}). 

\vspace{-10pt}
\section{Comparison With Existing Work}
\label{section:Comparison}
\vspace{-10pt}

Since the last decade, the possibility of lifting the metatheorems in \cite{Courcelle1990,ArnborgLagergrenSeese1991} to the 
directed setting has been an active line of research. 
 Indeed, following an approach delineated by Reeds~\cite{Reed1999} and Johnson, Robertson, Seymour and 
Thomas~\cite{JohnsonRobertsonSeymourThomas2001}, several digraph width measures have been defined in terms 
of the number of cops needed to capture a robber in a certain evasion game on digraphs. From these variations
we can cite for example, directed tree-width~\cite{Reed1999,JohnsonRobertsonSeymourThomas2001},
DAG width~\cite{BerwangerDawarHunterKreutzerObdrzalek2012}, D-width~\cite{Safari2005,Gruber2007}, 
directed path-width~\cite{Barat2006}, entanglement~\cite{BerwangerGradel2004,BerwangerGradelKaiserRabinovich2012}, 
Kelly width~\cite{HunterKreutzer2008} and Cycle Rank~\cite{Eggan1963,GruberHolzer2008}. 
All these width measures have in common the fact that DAGs have the lowest possible constant width ($0$ or $1$ depending on the measure). 
Other width measures in which DAGs do not have necessarily constant width include 
DAG-depth~\cite{GanianHlinenyKneisLangerObdrzRossmanith2009}, and Kenny-width~\cite{GanianHlinenyKneisLangerObdrzRossmanith2009}. 

The introduction of the digraph width measures listed above was often accompanied by algorithmic implications. 
For instance, certain linkage problems that are NP-complete for general graphs, e.g. Hamiltonicity, can be 
solved efficiently in graphs of constant directed tree-width \cite{JohnsonRobertsonSeymourThomas2001}. The winner 
of certain parity games of relevance to the theory of $\mu$-calculus can be determined efficiently in digraphs of 
constant DAG width \cite{BerwangerDawarHunterKreutzerObdrzalek2012}, while it is not known if the same can be 
done for general digraphs. Computing disjoint paths of minimal weight, a problem which is NP-complete in 
general digraphs, can be solved efficiently in graphs of bounded Kelly width. However, except for such sporadic 
successful algorithmic implications, researchers have failed to come up with an analog of Courcelle's theorem 
for graph classes of constant width for any of the digraph width measures described above. It turns out that 
there is a natural barrier against this goal: It can be shown that unless all the problems in the polynomial hierarchy have sub-exponential algorithms, 
which is a highly unlikely assumption, \msotwo model checking is intractable in any class of graphs that is closed under 
taking subgraphs and whose undirected tree-width is poly-logarithmic unbounded \cite{Kreutzer2012,KreutzerTazari2010}.
An analogous result can be proved with respect to model checking of \msoone properties if we assume a non-uniform version 
of the extended exponential time hypothesis \cite{GanianHlinenyLangerObdrzalekRossmanithSikdar2012,GanianHlinenyKneisMeisterObdrzalekRossmanithSikdar2010}. 
All classes of digraphs of constant width with respect to the directed measures described above are closed under subgraphs 
and have poly-logarithmically unbounded tree-width, and thus fall into the impossibility theorem of \cite{Kreutzer2012,KreutzerTazari2010}. 
It is worth noting that Courcelle, Makowsky and Rotics have shown that \msoone model checking is tractable in classes of graphs 
of constant clique-width \cite{CourcelleMakowskyRotics2000,CourcelleMakowskyRotics2001}, and that these classes are poly-logarithmic 
unbounded, but they are not closed under taking subgraphs. 

We define the {\em zig-zag number} of a digraph $G$ to be the minimum $z$ for which $G$ has a $z$-topological ordering, and 
denote it by $\zigzagnumber(G)$. The zig-zag number is a digraph width measure that is closed under taking subgraphs and 
that has interesting connections with some of the width measures described above. In particular we can prove the following theorem 
stating that families of graphs of constant zig-zag number are strictly richer than families of graphs of constant directed 
path-width. 

\begin{theorem}
\label{theorem:DirectedPathWidthZigZagNumber}
Let $G$ be a digraph of directed path-width $d$. Then $G$ has 
zig-zag number $z\leq 2d+1$. Furthermore, given a directed path decomposition 
of $G$ one can efficiently derive a $z$-topological ordering of $G$. On the other 
hand, there are digraphs on $n$ vertices whose zig-zag number is $2$ but whose 
directed path-width is $\Theta(\log n)$. 
\end{theorem}

Theorem \ref{theorem:DirectedPathWidthZigZagNumber} legitimizes the algorithmic 
relevance of Theorem \ref{theorem:MainTheorem} since 
path-decompositions of graphs of constant directed path-width can be computed in polynomial time \cite{Tamaki2011}. 
The same holds with respect to the cycle rank of a graph since constant cycle-rank 
decompositions\footnote{By cycle-rank decomposition we mean a direct elimination forest\cite{Gruber2012}.}
can be converted into constant directed-path decompositions in polynomial time \cite{Gruber2012}. 
Therefore all the problems described in Section \ref{section:Introduction} can be 
solved efficiently in graphs of constant directed path-width and in graphs of constant cycle rank.
We should notice that our main theorem circumvents the impossibility 
results of \cite{Kreutzer2012,KreutzerTazari2010,GanianHlinenyLangerObdrzalekRossmanithSikdar2012,GanianHlinenyKneisMeisterObdrzalekRossmanithSikdar2010}
by confining the monadic second order logic properties to subgraphs that are the union of $k$ {\em directed}
paths. 

A pertinent question consists in determining whether we can eliminate either $z$ or $k$ from the exponent 
of the running time $f(\varphi,k,z)\cdot n^{O(k\cdot z)}$ stated in Theorem \ref{theorem:MainTheorem}.
The following two theorems say that under strongly plausible parameterized complexity assumptions \cite{DowneyFellows1992}, 
namely that $W[2]\neq FPT$ and $W[1]\neq FTP$, the dependence of both $k$ and $z$ in the exponent of the 
running time is unavoidable. 

\vspace{-1pt}
\begin{theorem}[Lampis-Kaouri-Mitsou\cite{LampisKaouriMitsou2011}]
\label{theorem:LampisKaouriMitsou}
Determining whether a digraph $G$ of cycle rank $z$ has a Hamiltonian circuit is $W[2]$ hard with respect to 
$z$. 
\end{theorem}

Since by Theorem \ref{theorem:DirectedPathWidthZigZagNumber} constant cycle rank is less expressive than constant zig-zag number, the hardness result stated 
in Theorem \ref{theorem:LampisKaouriMitsou} also works for zig-zag number. 
 Given a sequence of $2k$ not necessarily distinct vertices  $\sigma=(s_1,t_1,s_2,t_2,...,s_k,t_k)$, a 
$\sigma$-linkage is a set of internally disjoint directed paths $p_1,p_2,...,p_k$ where each $p_i$ connects $s_i$ to $t_i$.
%
%The following theorem due to Slivkins says that under the assumption that the first level $W[1]$ of the $W$
%hierarchy does not collapses to $FTP$, there is no hope to remove the dependence of $k$ in the running 
%time $f(\varphi,z,k)\cdot n^{O(k\cdot z)}$ stated in Theorem \ref{theorem:MainTheorem}. Such a collapse 
%is considered higly unlikely\cite{DowneyFellows1992}.
%
%	Johnson, Robertson and Seymour proved that one can decide in polynomial time the 
%		existence of a $\sigma$-linkage in graphs of bounded {\bf directed} tree-width \cite{JohnsonRobertsonSeymourThomas2001}, 
%

\vspace{-1pt}
\begin{theorem}[Slivikins\cite{Slivkins2003}]
\label{theorem:Slivkins}
Given a $DAG$ $G$, determining whether $G$ has a $\sigma$-linkage $\sigma=(s_1,t_1,s_2,t_2,...,s_k,t_k)$
is hard for $W[1]$. 
\end{theorem} 
\vspace{-1pt}
Thus, since a $\sigma$-linkage is clearly the union of $k$-paths, Theorem \ref{theorem:Slivkins} implies that 
the dependence of $k$ on the exponent is necessary even if $z$ is fixed to be $1$.

%
%If $G$ is a digraph on $n$ vertices, then a directed Hamiltonian circuit on $G$ is an
%alternated sequence of vertices and edges $v_1e_1v_2e_2...v_ne_nv_1$  where 
%$e_i=(v_i,v_{i+1\!\!\mod\! n})$ and $v_i\neq v_j$ for $i\neq j$. 
%
%
%\begin{corollary}
%Under the assumption that $W[1]\neq FTP$ there are no computable functions $f,g$ such that the problem 
%stated in theorem \ref{theorem:MainTheorem} is solvable in time $f(\varphi,z,k)\cdot n^{g(k)}$. Under the 
%assumption that $W[2]\neq FTP$, there are no computable functions such that this problem is
%solvable in time $f(\varphi,z,k)\cdot n^{g(z)}$. 
%\end{corollary}
\vspace{-2pt}
\section{Zig-Zag Number versus Other Digraph Width Measures}
\label{section:zigzagVersusDirectedWidth}
\vspace{-2pt}
Cops-and-robber games provide an intuitive way to define several of the 
directed width measures cited in Section \ref{section:Comparison}.
Let $G$ be a digraph. A cops-and-robber game on $G$ is played by 
two parties. One is controlling a set of 
$k$ cops and the other is controlling a robber. At each round of the game 
the cop either stands on a vertex of $G$ or flies in an helicopter, meaning that it is temporarily removed from the 
digraph. The robber stands on a vertex of $G$, and can at any time 
run at great speed along a cop-free directed path to another vertex. 
The objective of the cops is to capture the robber by landing 
on a vertex currently occupied by him, while the objective of the robber is 
to avoid capture. 
%
%Let $\mu$ be the minimum number of cops needed to capture the robber in a digraph $G$.
%Then $\mu$ defines the directed tree-width of $G$ ($dtw(G)$) if 
%at each round the robber only moves to a vertex within the strongly connected component induced by the vertices that 
%are not blocked by cops \cite{JohnsonRobertsonSeymour2001}. $\mu$ defines the 
%$D$-width of $G$ ($Dw(G)$) if during the game the cops never revisit vertices that
%were previously occupied by some cop and if the robber move within the
%SCC induced by the non-blocked vertices\cite{Gruber2010}. $\mu$ defines the 
%DAG-width  of $G$ ($dagw(G)$) when the cops follow a monotone strategy but 
%when the robber is allowed to move outside the SCC component 
%induced by the non-blocked vertices \cite{BerwangerDawarHunterKreutzerObdrzalek2012}. 
%The quantity $\mu-1$ defines the directed path width of $G$ ($dpw(G)$) if 
%the robber is invisible and if it can move along any arbitrary cop free path \cite{Barat2006}.
%The Kelly-width of $G$ ($kellyw(G)$)  is the minimum number of cops needed to capture 
%the robber when the cops cannot see the robber but with the restriction that 
%the robber can only move when a cop is about to land in his current position\cite{GanianHlinenyKneisLangerObdrzRossmanith2009}.
%The DAG-depth of $G$ ($ddp(G)$) is the minimum number of cops needed to capture the robber when the cops follow a 
%lift-free strategy, i.e., a cop player never moves a cop from a vertex once it has landed \cite{GanianHlinenyKneisLangerObdrzRossmanith2009}. 

Let $\mu$ be the minimum number of cops needed to capture the robber in a digraph $G$.
The {\em directed tree-width} of $G$ ($dtw(G)$) is equal to $\mu$ if at each round the 
robber only moves to a vertex within the strongly connected component (SCC) induced by the vertices that 
are not blocked by cops \cite{JohnsonRobertsonSeymourThomas2001}. The {\em $D$-width} of $G$ ($Dw(G)$)
is equal to $\mu$ if the cops capture the robber according to a monotone strategy, i.e., 
cops never revisit vertices that were previously occupied by some cop. The robber is also required to move within the 
SCC induced by the non-blocked vertices\cite{Gruber2012}. The {\em DAG-width} of $G$ ($dagw(G)$)
is equal to $\mu$ when the cops follow a monotone strategy but when the robber can move along 
arbitrary cop free paths, independently of whether it stays within the SCC component 
induced by the non-blocked vertices \cite{BerwangerDawarHunterKreutzerObdrzalek2012}. 
The {\em directed path width} of $G$ ($dpw(G)$) is equal to the quantity $\mu-1$ if we add the additional 
complication that the cops cannot see the robber \cite{Barat2006}.
The {\em Kelly-width} of $G$ ($kellyw(G)$)  is equal to $\mu$ if 
the cops cannot see the robber and if at each step the robber only moves when a cop is about to land in his current position\cite{GanianHlinenyKneisLangerObdrzRossmanith2009}.
Finally, the {\em DAG-depth} of $G$ ($ddp(G)$), which is the directed analog of the tree-depth 
defined in \cite{NesetrilDeMendez2006}, is the minimum number of cops needed to capture the robber when the cops follow a 
lift-free strategy, i.e., the cop player never moves a cop from a vertex once it has landed \cite{GanianHlinenyKneisLangerObdrzRossmanith2009}.

Some other width measures are better defined via some structural property. For 
instance, the {\em K-width} of $G$ ($Kw(G)$) is the maximum number of different simple paths 
between any two vertices of $G$ \cite{GanianHlinenyKneisLangerObdrzRossmanith2009}. 
The weak separator number of $G$ is defined as follows: If $G=(V,E)$ is a digraph and 
$U\subseteq V$, then a weak balanced separator for $U$ is a set $S$ such that every 
SCC of $G[U\backslash S]$ contains at most $\frac{1}{2}|U|$ vertices. The weak separator number of 
$G$, denoted by $s(G)$ is defined as the maximum size taken over all subsets $U\subseteq V$, among the minimum weak balanced
separators of $U$. Finally, the {\em cycle rank} of a digraph $G=(V,E)$ denoted by $r(G)$ is
inductively defined as follows: If $G$ is acyclic, then $r(G)=0$. If $G$ is strongly connected and $E\neq \emptyset$, 
then $r(G) = 1 + \min_{v\in V}\{r(G-v)\}$. If $G$ is not strongly connected then $r(G)$ equals the maximum cycle 
rank among all strongly connected components of $G$.
Below we find a summary of the relations between the zig-zag number of a digraph and all the digraph width measures 
listed above. We write $A \precnsim B$ to indicate that there are graphs of constant width with respect 
to the measure $A$ but unbounded width with respect to the measure $B$. 
We write $A\preceq B$ to express that $A$ is not asymptotically greater than $B$. 

\begin{equation}
\label{equation:zigzagPathWidth}
zn(G)\precnsim dpw(G) \stackrel{\mbox{\tiny{\cite{Gruber2012}}}}{\precnsim} cr(G) \stackrel{\mbox{\tiny{\cite{GanianHlinenyKneisMeisterObdrzalekRossmanithSikdar2010}}}}{\precnsim} \left\{\begin{array}{l} Kw(G) \\ ddp(G) \end{array} \right. \hspace{1cm}\frac{cr(G)}{\log n} \stackrel{\mbox{\tiny{\cite{Gruber2012}}}}{\preceq} s(G)  
\end{equation}

\begin{equation}
\label{equation:zigzagDwidth}
\frac{zn(G)}{\log n} \preceq  s(G) \stackrel{\mbox{\tiny{\cite{Gruber2012}}}}{\preceq} Dw(G) \stackrel{\mbox{\tiny{\cite{Gruber2012}}}}{\preceq}  dagw(G) 
\stackrel{\mbox{\tiny{\cite{BerwangerDawarHunterKreutzerObdrzalek2012}}}}{\preceq} dpw(G)
\end{equation}

\begin{equation}
\label{equation:zigzagTreeWidth}
\sqrt{\frac{zn(G)}{\log n}}\preceq dtw(G) \stackrel{\mbox{\tiny{\cite{HunterKreutzer2008}}}}{\preceq} kellyw(G) \hspace{1cm} \sqrt{Dw(G)} \stackrel{\mbox{\tiny{\cite{EvansHunterSafari2007}}}}{\preceq} dtw(G) 
\stackrel{\mbox{\tiny{\cite{EvansHunterSafari2007}}}}{\preceq} Dw(G)
\end{equation}

The numbers above $\precnsim$ and $\preceq$ point to the references in which these relations where established.
The only new relations are $zn(G)\precnsim dpw(G)$, $zn(G)/\log n \preceq  s(G)$ and $\sqrt{zn(G)/\log n}\preceq dtw(G)$  
for which we provide a justification in the appendix.

\vspace{-10pt}
\section{Regular Slice Languages}
\label{section:RegularSliceLanguages}
\vspace{-5pt}

A slice $\boldS=(V,E,l,s,t,o)$ is a digraph comprising a set of vertices $V$, a set of edges $E$,
a vertex labeling function $l:V\rightarrow \labelAlphabet$ for some set of symbols $\Gamma$, and 
functions $s,t:E\rightarrow V$ which respectively associate to each edge $e\in E$,
a source vertex $e^s$ and a tail vertex $e^t$. We notice that an edge might possibly have the same source and tail ($e^s=e^t$). 
The vertex set $V$ is partitioned into three disjoint 
subsets: an in-frontier $I\subseteq V$ a center $C\subseteq V$ and an out-frontier $O\subseteq V$. 
Additionally, we require that each frontier-vertex in $I\cup O$ is the endpoint of exactly one edge in $E$ 
and that no edge in $E$ has both endpoints in the same frontier. The function $o:E\rightarrow \{-1,1\}$ 
is an orienting function with the restriction that $o(e)=1$ if $e^t\in O$, and 
$o(e)=-1$ if $e^t\in I$. Intuitively $o$ assigns $1$ to an edge if it is oriented towards
the out frontier and $-1$ if it is oriented towards the in-frontier. 
The frontier vertices in $I\cup O$ are labeled by $l$ with numbers from the set $\{1,...,q\}$ 
for some natural number $q\geq \max\{|I|,|O|\}$ in such a way that no two vertices in the same frontier receive the same number.
Vertices belonging to different frontiers may on the other hand be labeled with the same number. The 
center vertices in $C$ are labeled by $l$ with elements from $\Gamma\backslash \{1,...,q\}$. We say that 
a slice $\boldS$ is normalized if $l(I)=\{1,...,|I|\}$ and $l(O)=\{1,...,|O|\}$. Non-normalized slices
will play an important role in Section \ref{section:MainTheorem}. Since we will deal with weighted
graphs, we will also allow the edges of a slice to be weighted by a function $w:E\rightarrow \Omega$
where $\Omega$ is a finite commutative semigroup. 

A slice $\boldS_1$ with frontiers $(I_1,O_1)$ can be glued to a slice $\boldS_2$ with frontiers $(I_2,O_2)$
provided $l_1(O_1)=l_2(I_2)$ and provided that $o_1(e_1)=o_2(e_2)$ whenever $e_1^t\in O_1, e_2^s\in I_2$ and 
$l_1(e_1^t)=l_2(e_2^s)$. In that case the glueing gives rise to the slice $\boldS_1\circ \boldS_2$ with 
frontiers $(I_1,O_2)$ which is obtained by fusing each such a pair of edges $e_1,e_2$ into a single edge $e$
whose orientation is coherent with the orientation of $e_1$ and $e_2$\footnote{By coherent we mean $e^s=e_1^s$ and $e^t=e_2^t$ if $o(e_1)=1$ 
and $e^s=e_2^s$ and $e^t=e_1^t$ if $o(e_1)=-1$.} (Figure \ref{figure:BasicSlices}.$i$).
We observe that in the glueing process the frontier vertices disappear.
If $\boldS_1$ and $\boldS_2$ are weighted by functions $w_1$ and $w_2$, then
we add the requirement that the glueing of $\boldS_1$ with $\boldS_2$ can be 
performed if the weights of the edges touching the out-frontier of $\boldS_1$ 
agree with the weights of their corresponding edges touching the in-frontier of $\boldS_2$. 
On the other hand, any slice can be decomposed into a sequence of atomic parts which we call {\em unit slices}, namely, slices 
with at most one vertex on its center (Figure \ref{figure:BasicSlices}.$i$).
Thus slices may be regarded as a graph theoretic analog of 
the knot theoretic braids \cite{Artin1947}, in which twists are replaced by vertices.
Within automata theory, slices may be related to a vast number of formalisms such 
as graph automata \cite{Thomas1992,BrandenburgSkodinis2005}, 
graph rewriting systems
\cite{Courcelle1987,BauderonCourcelle1987,EngelfrietVereijken1997}, and others
\cite{GiammarresiRestivo1996,GiammarresiRestivo1992,BozapalidisKalampakas2006,BorieParkerTovey1991}. In 
particular, slices may be regarded as a specialized version of the multi-pointed
graphs defined in \cite{EngelfrietVereijken1997} but subject to a slightly different
composition operation. 

%\begin{figure}[!hb] 
%\centering 
%\includegraphics[scale=0.35]{BasicSlices1.eps}
%\caption{$i$) Composition of Slices. $ii)$ A unit decomposition of zig-zag number $3$. 
%The path $1-2-3-4$ has zig-zag number $1$ while the path $1-4-2-3$ has zig-zag number $3$ 
%$iii)$ A slice graph $\slicegraph$ $iv$) The graph language represented by $\slicegraph$. Following 
%the upper branch of $\slicegraph$ the generated graphs are cycles of size at least $3$ with a protuberance. 
%Following the lower branch, the generated graphs are directed lines of size at least $4$.
%} 
%\label{figure:BasicSlices}
%\label{figure:SliceGraphgraph}
%\end{figure}

The width of a slice $\boldS$ with frontiers $(I,O)$ is defined as $w(\boldS)=\max\{|I|,|O|\}$.
In the same way that letters from an alphabet may be concatenated by automata to form infinite 
languages of strings, we may use automata or regular expressions over alphabets
of slices of a bounded width to define infinite families of digraphs. Let $\slicealphabet^{c,q}$ denote 
the set of all unit slices of width at most $c$ and whose frontier vertices are numbered with numbers from $\{1,...,q\}$ for 
$q\geq c$. We say that a slice is {\em initial} if its in-frontier is empty and {\em
final} if its out-frontier is empty. A slice with empty center is called a {\em permutation 
slice}. Due to the restriction that each frontier vertex of a slice must be connected to 
precisely one edge, we have that each vertex in the in-frontier of a permutation slice is necessarily connected 
to a unique vertex in its out-frontier. The empty slice, denoted by $\emptyslice$, is the slice with empty center 
and empty frontiers. We regard the empty slice as a permutation slice. A subset $\lang$ of the free monoid
$(\slicealphabet^{c,q})^*$ generated by $\slicealphabet^{c,q}$ is a slice language if 
for every sequence of slices $\boldS_1\boldS_2...\boldS_n \in \lang$ we have 
that $\boldS_1$ is an initial slice, $\boldS_n$ a final slice and 
$\boldS_i$ can be glued to $\boldS_{i+1}$ for each $i\in \{1,...,n-1\}$.
We should notice that at this point the operation of the monoid in consideration is just the concatenation 
$\boldS_1\boldS_2$ of slice symbols $\boldS_1$ and $\boldS_2$ and should not be confused with the composition $\boldS_1 \circ \boldS_2$ 
of slices. Also the unit of the monoid is just the empty symbol $\lambda$ and not the empty 
slice, thus the elements of $\lang$ are simply sequences of slices, regarded as dumb letters. 
To each slice language
$\lang$ over $\slicealphabet^{c,q}$ we associate a graph language $\lang_{\graph}$ consisting 
of all digraphs that are obtained by composing the elements of the strings in $\lang$: 

%\vspace{-10pt}
\begin{equation}
\label{equation:SliceLanguage}
\lang_{\graph} = \{\boldS_1\circ\boldS_2...\circ \boldS_n| \boldS_1\boldS_2...\boldS_n \in \lang\}
\end{equation}

However we observe that a set $\lang_{\graph}$ of digraphs may be represented by several different slice 
languages, since a digraph in $\lang_{\graph}$ may be decomposed in several ways as a string of unit slices. 
We will use the term {\em unit decomposition} of a digraph $H$ to denote any sequence of unit slices
$\boldU=\boldS_1\boldS_2...\boldS_n$ whose composition $\boldS_1\circ\boldS_2\circ...\circ\boldS_n$ yields $H$.
We say that the unit decomposition $\boldU$ is {\em dilated} if it contains permutation slices, including possibly the 
empty slice (Figure \ref{figure:Subslicings}.$iii$). The slice-width of $\boldU$ is the minimal $c$ for which $\boldU\in (\slicealphabet^{c,q})^*$
for some $q$. In other words, the slice width of a unit decomposition is the width of the widest slice appearing in it.  

A slice language is regular if it is generated by a finite automaton or regular
expressions over slices. We notice that since any slice language is a subset of the free 
monoid generated by a slice alphabet $\slicealphabet^{c,q}$, we do not need to make
a distinction between regular and rational slice languages.
Therefore, by Kleene's theorem, every slice language generated by a regular expression can be also generated by a
finite automaton. Equivalently, a slice language is regular if and only if it
can be generated by the slice graphs defined below
\cite{deOliveiraOliveira2010}:

\begin{definition}[Slice Graph] A {\em slice graph} over a slice alphabet
$\slicealphabet^{c,q}$ is a labeled directed graph
$\slicegraph=(\mathcal{V},\mathcal{E},\mathcal{S},\mathcal{I},\mathcal{T})$ 
possibly containing loops but without multiple edges where $\mathcal{I}\subseteq \mathcal{V}$
is a set of initial vertices, $\mathcal{T}\subseteq \mathcal{V}$ a set of final vertices and the
function $\mathcal{S}:\mathcal{V}\rightarrow \slicealphabet^{c,q}$ satisfies the following conditions:
\begin{itemize}
	\item $\mathcal{S}(\slicegraphvertex)$ is a initial slice for every vertex $\slicegraphvertex$ in $\mathcal{I}$, 
	\item $\mathcal{S}(\slicegraphvertex)$ is final slice for every vertex $\slicegraphvertex$ in $\mathcal{T}$ and,
	\item $(\slicegraphvertex_1,\slicegraphvertex_2) \in \mathcal{E}$ implies that $\mathcal{S}(\slicegraphvertex_1)$ can be glued to $\mathcal{S}(\slicegraphvertex_2)$. 
\end{itemize}
\end{definition}

We say that a slice graph is {\em deterministic} if none of its vertices has two forward neighbors labeled with the same slice and if there
is no two initial vertices labeled with the same slice. In other words, in a deterministic slice graph no two distinct walks are labeled with the same sequence of slices.  
We denote by $\lang(\slicegraph)$ the slice language generated by $\slicegraph$, 
which we define as the set of all sequences slices $\mathcal{S}(\slicegraphvertex_1)\mathcal{S}(\slicegraphvertex_2)\cdots\mathcal{S}(\slicegraphvertex_n)$ 
where $\slicegraphvertex_1\slicegraphvertex_2\cdots \slicegraphvertex_n$ is a walk on $\slicegraph$ from an initial vertex to a final vertex.
We write $\lang_{\graph}(\slicegraph)$ for the language of digraphs derived from $\lang(\slicegraph)$. 

%\vspace{-7pt}
\section{Counting Subgraphs Specified by a Slice Language}
\label{section:MainTheorem}
%\vspace{-7pt}

A sub-slice of a slice $\boldS$ is a subgraph of $\boldS$ that is itself a slice. If $\boldS'$ is a 
sub-slice of $\boldS$ then we consider that the numbering in the frontiers of $\boldS'$ are inherited from the numbering 
of the frontiers of $\boldS$. Therefore, even if $\boldS$ is normalized, its sub-slices might not be.  If $\boldU = \boldS_1\boldS_2...\boldS_n$ is a unit decomposition 
of a digraph $G$, then a sub-unit-decomposition of $\boldU$ is a unit decomposition $\boldU'=\boldS_1'\boldS_2'...\boldS_n'$
of a subgraph $H$ of $G$ such that $\boldS_i'$ is a sub-slice of $\boldS_i$ for $1\leq i \leq n$. We observe that 
sub-unit-decompositions may be padded with empty slices. A unit decomposition $\boldU=\boldS_1\boldS_2...\boldS_n$
may have exponentially many sub-unit-decompositions of a given slice-width $c$. However, as we will state in Lemma \ref{lemma:SliceGraphSubgraphs}
the set of all such sub-unit decompositions of $\boldU$ may\! be represented\! by a slice\! graph of 
size polynomial in $n$. A normalized unit decomposition is a unit decomposition $\boldU=\boldS_1\boldS_2...\boldS_n$ 
such that $\boldS_i$ is a normalized slice for each $i\in \{1,...,n\}$. A slice language is normalized if all 
unit decompositions in it are normalized. A slice-graph is normalized  if all slices labeling its vertices 
are normalized. We notice that a regular slice language is normalized if and only if it is generated by 
a normalized slice graph.

\begin{lemma}
\label{lemma:SliceGraphSubgraphs}
Let $G$ be a digraph with $n$ vertices, $\unitdecomposition=\boldS_1\boldS_2...\boldS_n$ be a normalized unit 
decomposition of $G$ of slice-width $\widthG$, and let $c\in \N$ be such that $c\leq \widthG$. Then one 
can construct in time $n\cdot\widthG^{O(c)}$ an acyclic and deterministic slice graph 
$\subgraphsSlicegraph^c(\unitdecomposition)$ on $n\cdot \widthG^{O(c)}$ vertices whose slice language
$\lang(\subgraphsSlicegraph^c(\boldU))$ consists of all sub-unit-decompositions of $\boldU$ of slice-width at most $c$. 
\end{lemma}

%\begin{figure}[!hb] 
%\centering 
%\includegraphics[scale=0.37]{subslicesSliceGraphNew}
%\caption{ $i$) A digraph $G$ and a subdigraph $H$ of $G$ $ii)$ A normalized unit decomposition of $G$ $iii)$ An unormalized dilated unit decomposition of $H$. 
%$\boldS_3'$ is a permutation slice $iv)$ Some sub-slices in $\subslicesNumber(\boldS_2)$ and in $\subslicesNumber(\boldS_3)$ and how they are connected. 
%$v$) A normalized slice (leftmost slice) and its $3$-numbering expansion (all slices together).} 
%\label{figure:Subslicings}
%\label{figure:NumberingExpansion}
%\end{figure}
%

Let  $\omega=(v_1,v_2,...,v_n)$ be a linear ordering of the vertices of a digraph $H$.
We say that a dilated unit decomposition $\boldU = \boldS_1\boldS_2...\boldS_m$ of $H$ is compatible with $\omega$ 
if $v_i$ is the center vertex of $\boldS_{j_i}$ for each $i\in \{1,...,n\}$ and if $j_i > j_{i-1}$ for 
each $i\in \{1,...,n-1\}$ (observe that we need to use the subindex $j_i$ instead of simply $i$ because $\boldU$ is dilated and 
therefore some slices in $\boldU$ have no center vertex). Notice that 
for each ordering $\omega$ there might exist several unit decompositions of $H$ that are compatible to $\omega$.  
If $\omega$ is a $z$-topological ordering of a digraph $G$ and if $\boldU$ is a dilated unit decomposition 
of $G$ that is compatible with $\omega$, then we say that $\boldU$ has zig-zag number $z$. The zig-zag number 
of a slice language $\lang$ is the maximal zig-zag number of a unit decomposition in $\lang$. 
If a dilated unit decomposition $\boldU$ has zig-zag number $z$ then any of its sub-unit decompositions has zig-zag number at 
most $z$ (Proposition \ref{proposition:SliceWidthFacts}). Thus the zig-zag number
of $\lang(\subgraphsSlicegraph^c(\boldU))$ is at most $z$. 

\begin{proposition}
\label{proposition:SubUnitDecompositionZigZagNumber}
Let $\boldU$ be a unit decomposition of zig-zag number $z$. Then any sub-unit-decomposition in $\lang(\subgraphsSlicegraph^c(\boldU))$ 
has zig-zag number at most $z$. 
\end{proposition}

A slice language $\lang$ is $z$-dilated-saturated, if $\lang$ has zig-zag number at most $z$ and if for every digraph $H\in \lang_{\graph}$, every $z$-topological ordering
$\omega$ of $H$ and every dilated unit decomposition $\boldU$ of $H$ that is compatible with $\omega$ we have that $\boldU \in \lang$. 
We should {\bf emphasize} that the intersection of the graph languages generated by two slice graphs 
is not in general reflected by the intersection of their slice languages.
Indeed, it is easy to define slice languages $\lang,\lang'$ for which $\lang_{\graph}=\lang_{\graph}'$ but for which $\lang \cap \lang'=\emptyset$!
Additionally, a reduction from the Post correspondence problem \cite{Post1946} established 
by us in \cite{deOliveiraOliveira2010} implies that even determining whether 
the intersection of the graph languages generated by slice languages is 
empty, is undecidable. However this is not an issue if at least one of the 
intersecting languages is $z$-dilated-saturated, as stated in the next proposition.  

\begin{proposition}
\label{proposition:Intersection}
Let $\lang$ and $\lang'$ be two slice languages over $\slicealphabet^{c,q}$, such that 
$\lang$ has zig-zag number $z$ and such that $\lang'$ is $z$-saturated. 
If we let $\lang^{\cap}=\lang\cap \lang'$, then $\lang^{\cap}_{\graph} = \lang_{\graph}\cap \lang'_{\graph}$. 
\end{proposition}

If $\boldS$ is a normalized slice in $\slicealphabet^{c,q}$ with in-frontier $I$ and out-frontier $O$ then a {\em $q$-numbering}
of $\boldS$ is a pair of functions  $in:I\rightarrow \{1,...,q\}$, $out:O\rightarrow\{1,...,q\}$ such that
for each two vertices $v,v'\in I$, $l(v)<l(v')$ implies that $in(l(v))<in(l(v'))$ and, 
for each two vertices $v,v'\in O$, $l(v)<l(v')$ implies that $out(l(v))<out(l(v'))$. We let $(\boldS,in,out)$ 
denote the slice obtained from $\boldS$ by renumbering each frontier vertex $v\in I$ with $in(l(v))$ and 
each out frontier vertex $v\in O$ with the $out(l(v))$. The $q$-numbering-expansion of a normalized slice\! $\boldS$ 
is\! the set\! $\numberingExtension(\boldS)$\! of all\! $q$-numberings\! of\! $\boldS$.

Let $\slicegraph = (\mathcal{V},\mathcal{E},\mathcal{S},\mathcal{I},\mathcal{T})$ be a slice graph 
over $\slicealphabet^{c,q}$. Then the $\widthG$-numbering expansion of $\slicegraph$ is the slice graph
$\numberingExtension^{\widthG}(\slicegraph)= (\mathcal{V}',\mathcal{E}',\mathcal{S}',\mathcal{I}',\mathcal{T}')$ defined as follows. 
For each vertex $\slicegraphvertex \in \mathcal{V}$  and each slice $(\boldS(\slicegraphvertex),in,out) \in \numberingExtension^{\widthG}(\mathcal{S}(\slicegraphvertex))$
we create a vertex $\slicegraphvertex_{in,out}$ in $\mathcal{V}'$ and label it with $(\boldS,in,out)$. 
Subsequently we connect $\slicegraphvertex_{in,out}$ to $\slicegraphvertex'_{in',out'}$ if there was an edge $(\slicegraphvertex,\slicegraphvertex')\in \mathcal{E}$
and if $(\boldS,in,out)$ can be glued to $(\boldS',in',out')$. 

\begin{theorem}
\label{theorem:SliceLanguageSubgraphs}
Let $G$ be digraph, $\unitdecomposition = \boldS_1\boldS_2 ... \boldS_n$ be a normalized unit decomposition 
of $G$ of slice-width $\widthG$ and zig-zag number $z$, $\slicegraph$ be a normalized $z$-dilated-saturated slice graph over $\slicealphabet^{c,q}$ 
and $\numberingExtension^{\widthG}(\slicegraph)$ be the $\widthG$-numbering expansion of $\slicegraph$. 
Then the set of all sub-unit-decompositions of $\boldU$ of slice-width at most $c$ 
whose composition yields a graph isomorphic to some graph in $\lang_{\graph}(\slicegraph)$ is represented by the regular slice language
$\lang(\subgraphsSlicegraph^c(\boldU)) \cap \lang(\numberingExtension^{\widthG}(\slicegraph))$. 
\end{theorem}

Let $\slicegraph=(\mathcal{V},\mathcal{E},\mathcal{S},\mathcal{I},\mathcal{T})$ be a slice graph and $(\Omega,+)$ be a finite commutative
semigroup with an identity element $0$. Then 
the {\em $\Omega$-weight expansion} of $\slicegraph$ is the slice graph 
$\mathcal{W}^{\Omega}(\slicegraph)=(\mathcal{V}',\mathcal{E}',\mathcal{S}',\mathcal{I}',\mathcal{T}')$ defined as 
follows: For each vertex $\slicegraphvertex\in \mathcal{V}$ labeled with the slice $\mathcal{S}(\slicegraphvertex)=(V,E,l)$,
we add the set of vertices $\{\slicegraphvertex_{w,tot}\}_w$ to $\mathcal{V}'$ where 
$w$ ranges over all weighting functions $w:E\rightarrow \Omega$ and $tot$ ranges over $\Omega$. We label each 
$\slicegraphvertex_{w,tot}$ with the tuple $(\mathcal{S}(\slicegraphvertex),w,tot)$. Then we add an edge $(\slicegraphvertex_{w,tot}, \slicegraphvertex'_{w',tot'})$ to 
$\mathcal{E}'$ if and only if $(\slicegraphvertex, \slicegraphvertex')\in \mathcal{E}$, if the slice $(\mathcal{S}(\slicegraphvertex),w)$ 
can be glued to the slice $(\mathcal{S}(\slicegraphvertex'),w')$ and if $tot'=tot + \sum_{e\in E'^{out}} w(e)$.
The set of final vertices $\mathcal{T}'$ consists of all vertices in $\mathcal{V}'$ which are labeled with a triple $(\boldS,w,tot)$ 
where $\boldS$ is a final slice. The set of initial vertices $\mathcal{I}'$ consists of all vertices 
in $\mathcal{V}'$ which are labeled with a triple $(\boldS,w,0)$ where $\boldS$ is an initial slice. 
Intuitively if $\slicegraph$ generates a language of graphs $\lang_{\graph}$, then $\mathcal{W}^{\Omega}(\slicegraph)$ generates the 
language $\lang_{\graph}'$ of all possible weighted versions of graphs in $\lang_{\graph}(\slicegraph)$.
In Theorem \ref{theorem:MainTechnical} below $q$ is the cut-width of $G$ and therefore it can be as large 
as $O(n^2)$. The parameter $c$ on the other hand is the slice-width of the subgraphs that are being counted.

\vspace{-3pt}
\begin{theorem}[Subgraphs in a Saturated Slice Language]
\label{theorem:MainTechnical}
Let $G=(V,E)$ be a digraph of cut-width $\widthG$ with respect to a $z$-topological ordering $\omega=(v_1,v_2,...,v_n)$ of its vertices, 
and let $\slicegraph$ be a deterministic normalized $z$-dilated-saturated slice graph over $\slicealphabet^{c,q}$ on $\sizeSliceGraph$ vertices.
Let $w:E\rightarrow \Omega$ be an weighting function on $E$ and $l=O(n)$ be a number. 
Then we may count in time $\sizeSliceGraph^{O(1)}\cdot n^{O(c)}\cdot \widthG^{O(c)}$ the number of subgraphs of $G$ of size 
$l$, that are isomorphic to some subgraph in $\lang_{\graph}(\slicegraph)$ and have maximal weight.
\end{theorem}
\begin{proof}
Let $\boldU=\boldS_1\boldS_2...\boldS_n$ be any normalized unit decomposition that is compatible with $\omega$, i.e.,
such that $v_i$ is the center vertex of $\boldS_i$ for $i=1,...,n$. Clearly such a unit decomposition can be constructed in 
polynomial time in $n$. Since $\slicegraph$ is dilated saturated, by 
Theorem \ref{theorem:SliceLanguageSubgraphs} the set of all subgraphs of $G$ that are isomorphic to some digraph in $\lang_{\graph}(\slicegraph)$
is represented by the regular slice language $\lang(\subgraphsSlicegraph^c(U))\cap \lang(\mathcal{W}^{\Omega}(\numberingExtension^q(\slicegraph)))$.
By Lemma \ref{lemma:SliceGraphSubgraphs} $\subgraphsSlicegraph^c(U)$ has $n\cdot \widthG^{O(c)}$ vertices and can be constructed 
within the same time bounds. The numbering expansion $\numberingExtension^q(\slicegraph)$ of $\slicegraph$ has \\ $\binom{q}{O(c)}\cdot \sizeSliceGraph = \sizeSliceGraph\cdot q^{O(c)}$
vertices and can be constructed within the same time bounds. The $\Omega$-expansion  $\mathcal{W}^{\Omega}(\numberingExtension^q(\slicegraph))$ of 
$\numberingExtension^q(\slicegraph)$ has $|\Omega|^{O(c)}\cdot\sizeSliceGraph\cdot q^{O(c)}=n^{O(c)}\cdot \sizeSliceGraph \cdot q^{O(c)}$ 
vertices and can be constructed within the same time bounds. 
Let $\slicegraph^{\cap} = \mathcal{W}^{\Omega}( \numberingExtension^q(\slicegraph))  \cap \subgraphsSlicegraph^c(\boldU)$. Since $\slicegraph^{\cap}$
can be obtained by a product construction, it has $r\cdot n^{O(c)}\cdot  q^{O(c)}$ vertices. 
Since $\subgraphsSlicegraph^c(\boldU)$ is acyclic, $\slicegraph^{\cap}$ is also acyclic.
Therefore counting the subgraphs in $G$ isomorphic to some graph in $\lang_{\graph}(\slicegraph)$ amounts to counting the number of simple directed
paths from an initial to a final vertex in $\slicegraph^{\cap}$. Since we are only interested 
in counting subgraphs with $l$ vertices, we can intersect this acyclic slice graph with the 
slice graph $\slicegraph^l$ generating all unit decomposition over $\slicealphabet^{c,q}$ containing precisely $l$ unit
slices that are not permutation slices. Again the slice graph $\slicegraph^{\cap}\cap \slicegraph^l$ will 
be acyclic. Finally since we are only interested in counting maximal-weight subgraphs, we delete from 
$\mathcal{T}'$ those vertices labeled with triples $(\boldS,w,tot)$ in which $tot$ is not maximal. 
The label of each path from an initial to a final vertex in this last slice graph identifies unequivocally a subgraph 
of $G$ of size $l$ and maximal weight. By standard dynamic programming we can count the number of paths in a DAG from a set of initial vertices to a set of final vertices in time polynomial on the number of vertices of the DAG (Proposition \ref{proposition:CountingDAG}). 
Thus we can determine the number of $l$-vertex maximal-weight subgraphs of $G$ which are isomorphic 
to some digraph in $\lang(\slicegraph)$ in time $r^{O(1)}n^{O(c)} q^{O(c)}$. 
$\square$
\end{proof}

\vspace{-15pt}
\section{Subgraphs Satisfying a given MSO property}
\label{section:MSO}
\vspace{-7pt}

In this section we will only give the necessary definitions to state Lemma \ref{lemma:MonadicRegularSet} and 
Theorem \ref{theorem:MonadicTechnical}, which are crucial steps towards the proof of Theorem \ref{theorem:MainTheorem}.
For an extensive account on the monadic second order logic of graphs we refer the reader to the treatise 
\cite{CourcelleEngelfriet2012} (in special Chapters 5 and 6). As it is customary, we will represent a digraph $G$ by a relational structure $G=(V,E,s,t,l_V,l_E)$
where $V$ is a set of vertices, $E$ a set of edges, $s,t \subseteq E\times V$ are respectively the source 
and tail relations, $l_V\subseteq V\times \Sigma_V$ and $l_E\subseteq V\times \Sigma_E$ are respectively the vertex-labeling and 
edge-labeling relations. We give the following semantics to 
these relations: $s(e,v)$ and $t(e,v')$ are true if $v$ and $v'$  are respectively the source and the tail of the edge $e$;
$l_V(v,a)$ is true if $v$ is labeled with the symbol $a \in \Sigma_V$ while $l_E(e,b)$ is true if $e$ is labeled with the symbol 
$b\in \Sigma_E$. We always assume that $e$ is oriented from its source to its tail.
Let $\{x,y,z,z_1,y_1,z_1,...\}$ be an infinite set of first order variables and $\{X,Y,Z,X_1,Y_1,Z_1,...\}$ be an infinite set of 
second order variables. Then the set of $MSO_2$ formulas is the smallest set of formulas containing: 
\begin{itemize}
	\item the atomic formulas $x\in X$, $V(x)$, $E(x)$, $s(x,y)$, $t(x,y)$, $l_V(x,a)$ for each $a\in \Sigma_V$, 
		$l_E(x,b)$ for each $b\in \Sigma_E$, 
	\item the formulas $\varphi \vee \psi$, $\neg \varphi$, $\exists x.\varphi(x)$ and $\exists X.\varphi(X)$, where 
		$\varphi$ and $\psi$ are $MSO_2$ formulas. 
\end{itemize}

If $\mathcal{X}$ is a set of second order variables, and $G=(V,E)$ is a graph, then an interpretation 
of $\mathcal{X}$ over $G$ is a function $M:\mathcal{X}\rightarrow 2^{V}$ that assigns to each variable 
in $\mathcal{X}$ a subset of vertices of $V$. The semantics of a formula $\varphi(\mathcal{X})$
over free variables $\mathcal{X}$ being true on a graph $G$ under interpretation $M$ is the usual one. 
A sentence is a formula $\varphi$ without free variables. For a sentence $\varphi$ and a graph
$G$, if it is the case that $\varphi$ is true in $G$, then we say that $G$ 
satisfies $\varphi$ and denote this by $G\models \varphi$. Now we are in a position to state a crucial Lemma 
towards the proof of Theorem \ref{theorem:MainTheorem}. Intuitively it states that for any \msotwo formula $\varphi$ 
the set of all unit decompositions of a fixed width of digraphs satisfying $\varphi$ forms a regular set. 

\begin{lemma} 
\label{lemma:MonadicRegularSet} For any \msotwo sentence $\varphi$ over digraphs and any $c\in \N$, 
the set $\lang(\varphi,\slicealphabet^c)$ of all slice strings $\boldS_1\boldS_2...\boldS_k$ 
over $\slicealphabet^c$ such that $\boldS_1\circ \boldS_2 \circ ...\circ \boldS_k = G$ and $G\models \varphi$ is a 
regular subset of $(\slicealphabet^c)^*$.  
\end{lemma}

Lemma \ref{lemma:MonadicRegularSet} gives a slice theoretic 
analog of Courcelle's model checking theorem: In order to verify whether a digraph $G$ of existential slice-width at most $c$ satisfies
a given MSO property $\varphi$, one just needs to find a slice decomposition $\boldU=\boldS_1\boldS_2...\boldS_n$ of $G$ and 
subsequently verify whether the deterministic finite automaton (or slice graph) accepting $\lang(\varphi, \slicealphabet^{c})$ accepts $\boldU$.
However the goal of the present work is to make a rather different use of Lemma \ref{lemma:MonadicRegularSet}. Namely, next in Theorem \ref{theorem:MonadicTechnical} 
we will restrict Lemma \ref{lemma:MonadicRegularSet} in such a way that it concerns only $z$-saturated regular slice languages, so that 
it can be coupled to Theorem \ref{theorem:MainTechnical}, yielding in this way a proof of our main theorem (Theorem \ref{theorem:MainTheorem}).
%The main theorem of this section establishes
%a connection between $MSO_2$ and $z$-dilated-saturated regular slice languages. 
%
%\vspace{-5pt}
\begin{theorem}
\label{theorem:MonadicTechnical}
For any \msotwo formula $\varphi$ and any $k,z\in \N$, one may effectively construct a $z$-dilated-saturated slice graph $\slicegraph(\varphi,k,z)$ 
over the slice alphabet $\slicealphabet^{k\cdot z}$ whose graph language $\lang_{\graph}(\slicegraph(\varphi,k,z))$ consists precisely of the 
digraphs of zig-zag number at most $z$ that satisfy $\varphi$ and that are the union of $k$ directed paths. 
\end{theorem}
%\vspace{-5pt}

Finally we are in a position to prove Theorem \ref{theorem:MainTheorem}. The proof will follow 
from a combination of Theorems \ref{theorem:MonadicTechnical} and \ref{theorem:MainTechnical}.

\paragraph{\bf Proof of Theorem \ref{theorem:MainTheorem}}
Given a monadic second order formula $\varphi$, and positive integers $k$ and $z$, first we construct
the dilated-saturated slice graph $\slicegraph(\varphi,z,k)$ over $\slicealphabet^{k\cdot z}$ 
as in Theorem \ref{theorem:MonadicTechnical}. Since the slice-width of a digraph is at most $O(n^2)$
if we plug $q=O(n^2)$, $r=|\slicegraph(\varphi,z,k)|$ and $\slicegraph(\varphi,z,k)$ into Theorem \ref{theorem:MainTechnical}, 
and if we let $f(\varphi,z,k)=r^{O(1)}$, then we get an overall upper bound 
of $f(\varphi,z,k)\cdot n^{O(k\cdot z)}$ for computing the number of subgraphs of $G$ that satisfy
$\varphi$ and that are the union of $k$-directed paths.
$\square$

\vspace{-10pt}
\section{Final Comments}
\label{section:Conclusion}
\vspace{-10pt}

In this work we have employed slice theoretic techniques
to obtain the polynomial time solvability of many natural 
combinatorial questions on digraphs of constant directed 
path-width, cycle rank, K-width and DAG-depth. 
We have done so by using the zig-zag number of a digraph as a 
point of connection between these directed width measures, 
regular slice languages and the monadic second order logic of graphs.
Thus our results shed new light into a field that has 
resisted algorithmic metatheorems for more than a decade.
More precisely, we showed that despite the severe restrictions imposed by the
impossibility results in \cite{Kreutzer2012,KreutzerTazari2010,GanianHlinenyLangerObdrzalekRossmanithSikdar2012,GanianHlinenyKneisMeisterObdrzalekRossmanithSikdar2010},
it is still possible to develop logical-based algorithmic
frameworks that are able to represent a considerable variety of 
interesting problems.

\vspace{-15pt}
\section{Acknowledgements}
\label{section:Acknowledgements}
\vspace{-10pt}
The author would like to thank Stefan Arnborg for interesting discussions
about the monadic second order logic of graphs and for providing valuable 
comments and suggestions on this work. 
%\vspace{-10pt}

\bibliographystyle{abbrv}
\bibliography{subgraphsSatisfyingMSOPropertiesOnZTopologicallyOrderableDigraphs-IPEC}

\newpage
\appendix

\begin{figure}[!hb] 
\centering 
\includegraphics[scale=0.35]{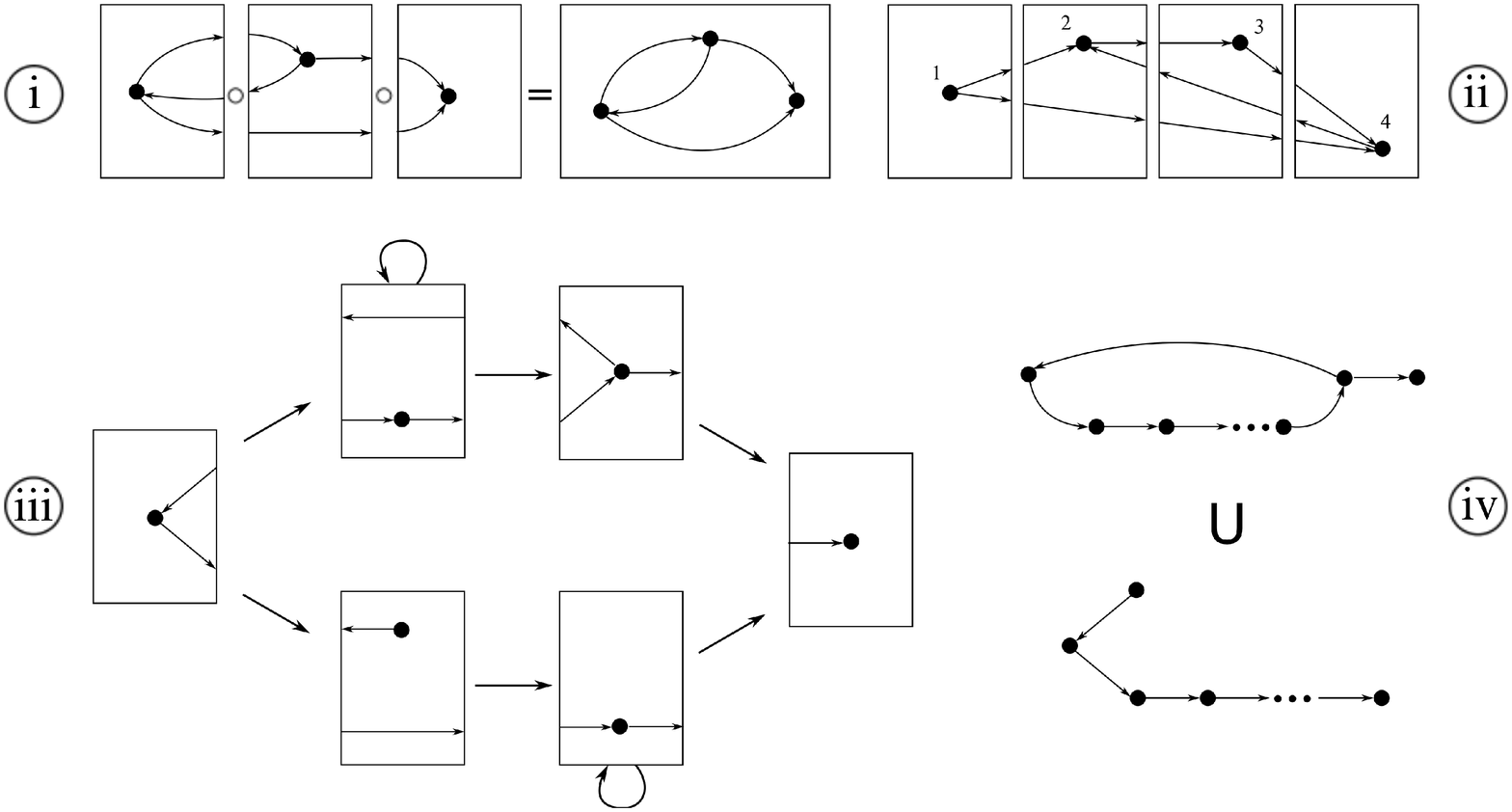}
\caption{$i$) Composition of Slices. $ii)$ A unit decomposition of zig-zag number $3$. 
The path $1-2-3-4$ has zig-zag number $1$ while the path $1-4-2-3$ has zig-zag number $3$ 
$iii)$ A slice graph $\slicegraph$ $iv$) The graph language represented by $\slicegraph$. Following 
the upper branch of $\slicegraph$ the generated graphs are cycles of size at least $3$ with a protuberance. 
Following the lower branch, the generated graphs are directed lines of size at least $4$.
} 
\label{figure:BasicSlices}
\label{figure:SliceGraphgraph}
\end{figure}

\begin{figure}[!hb] 
\centering 
\includegraphics[scale=0.37]{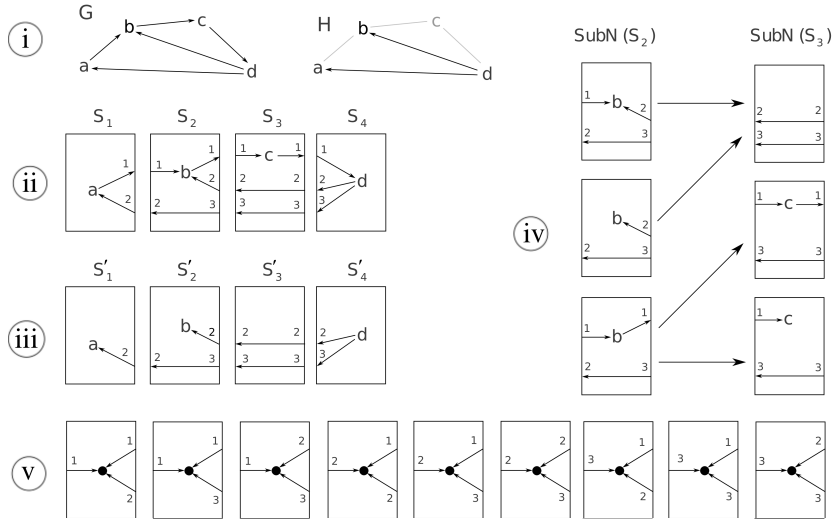}
\caption{ $i$) A digraph $G$ and a subdigraph $H$ of $G$ $ii)$ A normalized unit decomposition of $G$ $iii)$ An unnormalized dilated unit decomposition of $H$. 
$\boldS_3'$ is a permutation slice $iv)$ Some sub-slices in $\subslicesNumber(\boldS_2)$ and in $\subslicesNumber(\boldS_3)$ and how they are connected. 
$v$) A normalized slice (leftmost slice) and its $3$-numbering expansion (all slices together).} 
\label{figure:Subslicings}
\label{figure:NumberingExpansion}
\end{figure}

\newpage

\section{Proof of Theorem \ref{theorem:DirectedPathWidthZigZagNumber}}
\label{section:ZigZagNumber-OtherDigraphMeasures}

Let $G=(V,E)$ be a digraph. Then a directed path decomposition of $G$
is a sequence $\mathcal{P}=X_1,X_2,...,X_p\subseteq V$ of subsets of vertices
of $G$ such that $i$) $\cup_{j=1}^p X_j = V$, $ii$) for 
every $i,j,k$ with $i< j < k$, $X_i\cap X_k \subseteq X_j$ and 
$iii)$ for every directed edge $(u,v)\in E$, 
there exists a pair $i,j$ of indexes with 
$i\leq j$ such that $u\in X_i$ and $v\in X_j$. The width of 
$\mathcal{P}$, denoted $dpw(G,\mathcal{P})$ is the size 
of the largest set in $\mathcal{P}$. The directed path-width of 
$G$, denoted $dpw(G)$ is the minimal value of $dpw(G,\mathcal{P})$
where $\mathcal{P}$ ranges over all path decompositions of $G$. 
In this section we will prove the two claims made in Theorem \ref{theorem:DirectedPathWidthZigZagNumber}. 
The first, stating that constant directed path-width implies constant zig-zag number ({\bf Part I}), and the second 
stating that there are graphs of constant zig-zag number but unbounded {\bf directed} path-width ({\bf Part II}). 

Before proving the first part of Theorem \ref{theorem:DirectedPathWidthZigZagNumber}, we will first define 
another digraph measure, the {\em directed vertex separation number} (d.v.s.n.) of a digraph. 
Let $\omega=(v_1,v_2,...,v_n)$ be a linear ordering of the vertices
of a digraph $G$. Then the directed vertex separation number of $G$ with 
respect to $\omega$, denoted by $dvsn(G,\omega)$,  is the maximal number of vertices in $\{v_i,v_{i+1},...,v_n\}$ 
that have a successor in $\{v_1,...,v_{i-1}\}$ for some $i$.

\begin{equation}
\label{equation:DirectedVertexSeparationNumber}
dvsn(G,\omega) = \max_{i}|\{ v_j | (v_j,v_k)\in E, j\geq i, k<i  \}|
\end{equation}

The directed vertex separation number of a digraph $G$, denoted by $dvsn(G)$, is the minimal value of $dvsn(G,\omega)$ among all 
linear orderings of the vertex of $G$. It can be shown that the 
d.v.s.n. of a graph is equal to its directed path-width \cite{YangCao2008}, and 
that given a linear ordering $\omega$ of $G$ of d.v.s.n. equal to $d$, one 
can construct efficiently a directed path decomposition of $G$ of width $d$, and vice versa. 
Additionally, given a positive integer $d$, one may determine in time $O(mn^{d+1})$ whether $G$ 
has an ordering $\omega$ satisfying $dvsn(G,\omega)\leq d$, and in case it exists, 
return it in the same amount of time \cite{Tamaki2011}. Therefore to prove that 
bounded path-width implies bounded zigzag number, it is enough to show that any ordering of $G$ of $dvsn(G)=d$ has zig-zag number at most $2d+1$. 

\paragraph{{\bf Proof of Theorem \ref{theorem:DirectedPathWidthZigZagNumber} - Part I}}
Let $G$ be a digraph of directed path-width $d$. Then it follows from \cite{YangCao2008}
that there exists a linear ordering $\omega=(v_1,...,v_n)$ of the vertices of $G$ such that 
of direct vertex separation number $d$. Therefore for any $i\in \{1,...,n-1\}$ we have that 
there are at most $d$ vertices $v_{j_1},...,v_{j_d}$ in $V[i..n]=\{v_i,...,v_n\}$ which are the source of an edge with target in 
$V[1,...,i-1]=\{v_1,...,v_{i-1}\}$. This implies that for any path 
$p$ of $G$ there are at most $d$ edges of $p$ going from $V[i,...,n]$ to $V[1,...,n-1]$. But this by its turn implies that there are at 
most $d+1$ edges of $p$ going in the opposite direction, from $V[1,...,n-1]$
to $V[i,...,n]$. Therefore the cut width of $p$ w.r.t. $\omega$ is 
at most $2d+1$. $\square$

\paragraph{\bf Proof of Theorem \ref{theorem:DirectedPathWidthZigZagNumber} - Part II}

\begin{lemma}[\cite{Barat2006}]
\label{lemma:Barat2006}
Let $G$ be an undirected graph, and let $D$ be the digraph obtained 
from $G$ by replacing every undirected edge $\{u,v\}$ with two anti-parallel
directed edges $(u,v)$ and $(v,u)$. Then the {\bf directed} path-width of $D$ 
is equal to the {\bf undirected} path-width of $G$. 
\end{lemma}

The complete binary tree $T(n)$ is the binary tree 
on $n$ vertices in which every level, except possibly
the last is completely filled, and all nodes are as far 
left as possible. In the next Lemma we show that the directed
version of $T(n)$ has zig-zag number $2$ but has unbounded 
directed path-width. 

\begin{lemma}
Let $T(n)$ be the complete binary tree on $n$ vertices, and 
let $D(n)$ be the digraph obtained from $T(n)$ by replacing each 
of its undirected edges by a pair of directed edges of opposite directions.
Then $D(n)$ has zig-zag number $zn(D(n)) \leq  2$, while it has directed 
path-width $dpw(D(n))=\Omega(\log n)$. 
\end{lemma}
\begin{proof}
It is well known that the complete binary tree on $n$ vertices 
has {\bf undirected} path-width $\Theta(\log n)$, and indeed this 
tight bound follows from a characterization of path-width of a graph in terms
of the number of cops needed to capture an invisible robber on it. Applying Lemma 
\ref{lemma:Barat2006} we have that the digraph $D(n)$ derived from $T(n)$ 
has {\bf directed} path-width $\Theta(\log n)$. On the other hand 
we will show that $D(n)$ has zig-zag number at most $2$ for any $n$. Let 
$\omega=(v_1,v_2,...,v_n)$ be the ordering that traverses 
the vertices of $D(n)$ according to a depth first search starting 
from the root. In other words, if $v$ is a vertex of 
$D(n)$, $v^L$ belongs to the left subtree of $v$, and $v^{R}$ belongs
to the right subtree of $v$, then 
$\omega(v) < \omega(v^L) < \omega(v^R)$ (Figure \ref{figure:BinaryTree}). 
We show that $D(n)$ has zig-zag number at most $2$ with respect to $\omega$. Notice that since $D(n)$ has 
the structure of a tree, for any two vertices $v_1,v_2$ there is 
a unique simple path from $v_1$ to $v_2$. Now let $v$ be the minimal vertex such that 
both $v_1$ and $v_2$ are in the subtree rooted on $v$ (Observe that $v$ might be potentially 
equal to $v_1$ or to $v_2$) Then the path from $v_1$ to $v_2$ has 
necessarily to pass trough $v$. Given that $\omega$ is a depth first ordering, 
each each of the paths $p_1$ and $p_2$ from $v_1$ to $v$ and from $v$ to $v_2$ respectively, 
has zig-zag number at most $1$ with respect to $\omega$ (it can have zig-zag number $0$ if $v=v_1$ or $v=v_2$). 
Therefore the path $p_1\cup p_2$ from from $v_1$ to $v_2$ has zig-zag number at most $2$. 
$\square$
\end{proof}

\begin{figure}[!hb] 
\centering 
\includegraphics[scale=0.45]{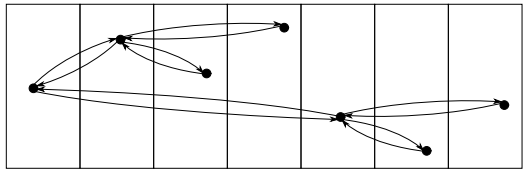}
\caption{ The digraph $D(7)$ obtained from the complete binary tree on $7$ vertices $T(7)$. The ordering 
depicted in the figure has zig-zag number $2$. } 
\label{figure:BinaryTree}
\end{figure}

%\section{Relationship between cut-width and $z$-topological orderings.}
%\label{section:CutWidthTopologicalOrderings}
%

\section{Proofs of Statements from Sections \ref{section:Introduction},\ref{section:Comparison} and \ref{section:zigzagVersusDirectedWidth}}
\label{section:ProofsLowSections} 

\paragraph{\bf Proof of equations \ref{equation:zigzagPathWidth}, \ref{equation:zigzagDwidth} and \ref{equation:zigzagTreeWidth} }
The relation  $zn(G) \precnsim dpw(G)$ comes from our Theorem \ref{theorem:DirectedPathWidthZigZagNumber}. The fact 
that $zn(G)/\log n \preceq s(n)$ is implied by $cr(G)/\log n \preceq s(G)$ together with $zn(G) \precnsim cr(G)$.
Finally $\sqrt{zn(G)/\log n} \preceq dtw(G)$  is implied by $\sqrt{Dw(G)} \leq dtw(G)$ together with $zn(G)/\log n \preceq Dw(G)$.
We observe that we are not aware of a generalization of the well known inequality $pw(G)\leq tw(G)\cdot O(\log n)$, relating
{\bf undirected} path-width to {\bf undirected} tree-width, to the directed setting. In other words we do not know whether $dpw(G)\leq dtw(G)\cdot O(\log n)$. 
Such a generalization would imply $zn(G)/log(n)\preceq dtw(G)$. $\square$

\paragraph{\bf Proof of Corollary \ref{corollary:StructuralCounting}} 
We start from item $10$: Robertson and Seymour's graph minor Theorem states that any minor closed graph property $\mathcal{P}$ 
is characterized by a finite set $\mathcal{F}$ of forbidden minors \cite{FriedmanRobertsonSeymour1987}.
The non-existence of each minor in $\mathcal{F}$ can be expressed by a monadic second order formula $\varphi_{\mathcal{F}}$ (See for example 
\cite{CourcelleEngelfriet2012});
The properties in items $2,4,5,6,7,8,9$ are all closed under minors, and therefore, by Item $10$ there is a monadic 
second order formula expressing each of these properties. $1$: Clearly there exists a \msotwo formula expressing 
connectedness, namelly, that there is a path between any two vertices. 
$3:$ A graph is bipartite if and only if it has no cycle of odd length, which is also easy to express in MSOL. $\square$

\section{Proofs of Results from Section \ref{section:MainTheorem}}
\label{section:LemmasAndPropositions}

\begin{proposition}
\label{proposition:SliceWidthFacts}
Let $H=(V,E)$ be a digraph, and $\omega= (v_1,v_2,...,v_n)$ be a $z$-topological ordering of $H$ of cut width $c$.
\begin{enumerate}
	\item \label{item:ZigZagNumber} If $H'$ is a subgraph of $H$ and $\omega'$ is the restriction of 
		$\omega$ to the vertices of $H'$, then 
		\begin{enumerate}
			\item\label{item:ZigZagNumberA}  $H'$ has cut-width at most $c$ w.r.t. $\omega'$. 
			\item\label{item:ZigZagNumberB} $\omega'$ is a $z$-topological ordering of $H'$. 
		\end{enumerate}
	\item \label{item:ZigZagNumberLessThanCutWidth} The zig-zag number $z$ is at most the cut-width $c$.  
	\item \label{item:ZigZagNumberSubaddtivity}
		If $H'=(V,E')$ is any digraph with same set of vertices as $H$ and if $H'$ has cut-width $c'$ with respect 
		to $\omega$ then $H\cup H'$ has cut-width at most $c+c'$ with respect to $\omega$. 
	\item \label{item:ZigZagNumberPathCovering}
	 	Let $\{p_1,p_2,...,p_k\}$ be a set of not necessarily edge disjoint nor vertex disjoint paths of $H$ such that $H=\bigcup_{i=1}^{k}p_i$. 
		Then $c\leq k\cdot z$. 
\end{enumerate}
\end{proposition}
\begin{proof}
%{\bf Proof of Proposition \ref{proposition:SliceWidthFacts}}
{\bf 1a)} Since $H$ has cut width $c$ with respect to $\omega$, we have that there are at most 
$c$ edges with one endpoint in $\{v_1,...,v_i\}$ and other endpoint in $\{v_{i+1},...,v_{n}\}$ for 
each $i$ with $1\leq i\leq n-1$. Therefore there are at most $c$ edges with one endpoint 
in $V(H)\cap \{v_1,...,v_i\}$ and other endpoint in $\{v_{i+1},...,v_{n}\}$ for each such an $i$. Implying 
in this way that $H$ has cut-width at most $c$ w.r.t the ordering $\omega'$ induced by $\omega$. 
{\bf 1b)} Since $\omega$ is a $z$-topological ordering of $H$, any path of $H$ has cut-width at most 
$z$ w.r.t. the restriction $\omega_{p}$ of $\omega$ to the vertices of $p$. Since any path $p'$ of
$H'$ is a sub-path of some path $p$ of $H$, by item 1.a $p'$ has cut-width at most $z$ with respect 
to the ordering $\omega_{p'}$ induced by $\omega_{p}$ (and consequently by $\omega'$) on the vertices of $p'$.
{\bf 2)} The proof is by contradiction. Suppose $z> c$. Then there exists a path $p$ in $H$ that has 
cut-width $z$ with respect to the restriction $\omega'$ of $\omega$ to the vertices of $p$. Since 
$p$ is a subgraph of $H$, by item $1a$, $z<c$ contradicting the assumption that $z>c$. {\bf 3)} 
For any $i$ such that $1\leq i \leq n-1$, there exist at most $c$ edges of $H$ with one endpoint in 
$\{v_1,...,v_i\}$ and other endpoint in $\{v_{i+1},...,v_n\}$. Analogously there are at most $c'$ 
edges of $H'$ with one endpoint in $\{v_1,...,v_i\}$ and other endpoint in $\{v_{i+1},...,v_n\}$. Therefore
there are at most $c+c'$ edges of $H\cup H'$ with one endpoint in $\{v_1,...,v_i\}$ and other endpoint in 
$\{v_{i+1},...,v_{n}\}$ for each such an $i$. Thus $H\cup H'$ has cut-width at most $c+c'$. {\bf 4)} 
Since $\omega$ is a $z$-topological ordering of $H$, each path $p$ of $H$ has cut-width at most $z$
with respect to $\omega$. Since $H=\cup_{i=1}^k p_i$ for paths $p_1,...,p_k$, by item 3 of this proposition,  
$c\leq k\cdot z$. $\square$
\end{proof}

%\subsection{Counting And Sampling Paths on DAGs}
%\label{section:CountingAndSamplingPathsOnDAGS}
%
\begin{proposition}[Counting Paths in a DAG]
\label{proposition:CountingDAG}
Let $\mathcal{D}=(V,E)$ be a DAG on $n$ vertices and let $S$ and $T$ be two arbitrary 
subsets of $V$. Then one may count the number of distinct simple paths 
from $S$ to $T$ in time $O(n^3)$. 
%Additionally, one may sample uniformly over such paths in polynomial time. 
\end{proposition}
\begin{proof}
%{\bf Proof of Proposition \ref{proposition:CountingDAG}}
We start by adding a vertex $v_0$ to $\mathcal{D}$ and connecting it to each vertex in $S$. 
Similarly, we add a vertex $v_{f}$ to $G$ and add a directed edge from each vertex in $T$ 
to $v_{f}$. We will assign weights for each vertex of $G$. For each $v\in V\cup \{v_0,v_f\}$,
the weight $w(v)$ will count the number of paths from $v$ to $v_{f}$. Indeed we will define a 
sequence of functions $w_0,w_1,...,w_h$ which will converge to $w$ when $h$ is greater or equal
to the height of $\mathcal{D}$, i.e., the size of the longest path from $v_0$ to $v_f$. Initially, 
we will set $w_0(v_{f})=1$ and $w_{0}(v)=0$ for all other vertices in $V\cup \{v_0\}$. 
Subsequently, for each $v\in V\cup\{v_0\}$, let $N^+(v)$ be the set of all forward neighbors of $v$. Then 
set $w_i(v) =\sum_{v'\in N^+(v)} w_{i-1}(v')$ for $i\geq 1$. By induction on the height of $\mathcal{D}$, 
we have that for each $v\in V\cup \{v_0\}$, the weight $w_h(v)$ represents precisely the number of simple 
paths from $v$ to $v_f$. In particular $w_h(v_0)$ represents the number of simple paths from $v_0$ to $v_f$ 
which is equal to the number of simple paths from $S$ to $T$. Therefore, let $w(v)=w_h(v)$ for each 
$v\in V$. 
%In order to sample uniformly from the set of all these simple paths, we proceed as follows: 
%We start by selecting vertex $v_0$. Subsequently at each step, if the last selected vertex was $v$, 
%then we choose a neighbor $v'$ of $v$ with probability $w(v')/w(v)$. In this way, each path $v_0v_1...v_k$
%will be selected with probability $\frac{w(v_1)}{w(v_0)}\cdot \frac{w(v_2)}{w(v_1)}\cdot ...\cdot \frac{w(v_f)}{w(v_k)} = \frac{w(v_f)}{w(v_0)} = \frac{1}{w(v_0)}$,
%or in other words, uniformly among the set of all paths from $v_0$ to $v_f$. $\square$
\end{proof}

\paragraph{\bf Proof of Proposition \ref{proposition:Intersection}}
The inclusion $\lang_{\graph}^{\cap}\subseteq \lang_{\graph}\cap \lang_{\graph}'$ holds for any two slice 
languages $\lang$ and $\lang'$ irrespectively of whether they are saturated or not: Let $H$ be a digraph in 
$\lang^{\cap}_{\graph}$. Then $H$ has a unit decomposition $\boldU=\boldS_1\boldS_2...\boldS_n$ 
in $\lang\cap \lang'$. Since $\boldU\in \lang$, $H\in \lang_{\graph}$ and, since $\boldU\in \lang'$, $H\in \lang_{\graph}'$.
Thus $\lang_{\graph}^{\cap} \subseteq \lang_{\graph}\cap \lang'_{\graph}$. Now we prove that if $\lang'$
is saturated the converse inclusion also holds: Let $H$ be a digraph in $\lang_{\graph}\cap \lang'_{\graph}$. Since 
$\lang$ has zig-zag number $z$, there exists a unit decomposition $\boldU=\boldS_1\boldS_2...\boldS_n$ of $H$ 
of zig-zag number $z$ in $\lang$. Since $\lang'$ is $z$-saturated any unit-decomposition of $H$ of 
zig-zag number at most $z$ is in $\lang'$, and in special $\boldU\in \lang'$. Therefore 
$\boldU\in \lang^{\cap}=\lang\cap \lang'$ and $H\in \lang_{\graph}^{\cap}$. $\square$

\paragraph{\bf Proof of Lemma \ref{lemma:SliceGraphSubgraphs}}
%As a first step, we will number the edges each slice $\boldS_i$ in $\boldU$ with the numbers of the frontier-vertices they touch in $\boldS_i$.
%To be more precise, if we let $l_{E_i}$ be the edge-labeling function of $\boldS_i$, then we will add the color $(in,j)$ to $l_{E_i}(e)$
%whenever $e$ touches the $j$-th in-frontier of $\boldS_i$. Similarly, we will add the color $(out,j)$ to $l_{E_i}(e)$ whenever $e$ touches
%the $j$-th out-frontier vertex of $\boldS_i$. In Figure \ref{figure:Subslicings}.$ii$ we illustrate this numbering process, where we 
%write only the number $j$ near each $j$-th frontier vertex.
%Subsequently, 
The slice graph $\subgraphsSlicegraph^c(\boldU)=(\mathcal{V},\mathcal{E},\mathcal{S},\mathcal{I},\mathcal{T})$ is constructed as follows: 
For each slice $\boldS_i\in \boldU$ (recall again that $\boldS_i$ belongs to $\slicealphabet^{\widthG}$ not in $\slicealphabet^c$, and that $\widthG \geq c$) 
we let  $\subslicesNumber^c(\boldS_i)$ be the set of all numbered sub-slices of $\boldS_i$
of slice-width at most $c$ (Figure \ref{figure:Subslicings}.$iv$), including slices with empty in-frontier, empty out-frontier 
or both (this last case embraces both the empty slice and slices with a unique vertex in the center and no frontier vertex).    
We should pay attention to the fact that the numbering of the frontier vertices of each such a subslice is inherited 
from the numbering of $\boldS_i$, as illustrated in figure \ref{figure:Subslicings}.$iii$, and thus these subslices are 
not necessarily normalized. This observation is crucial and will play a role in the fact that $\subgraphsSlicegraph^c(\boldU)$
is deterministic. For each such a sub-slice $\boldS\in \subslicesNumber^c(\boldS_i)$ (now $\boldS\in \slicealphabet^c$) 
we add a vertex $\slicegraphvertex_{i,\boldS}$ to $\mathcal{V}$ and label it with $\boldS$ (i.e., $\mathcal{S}(\slicegraphvertex_{i,\boldS})=\boldS$ ).
Subsequently we add an edge $(\slicegraphvertex_{i,S},\slicegraphvertex_{j,S'})$ if and only if $j=i+1$ and if $\boldS$ can be glued to 
$\boldS'$ respecting the numbering of the touching-frontier vertices. Observe that a slice with empty out-frontier can
always be glued to a slice with an empty in-frontier. This last observation allows us to represent some unit decompositions 
of disconnected sub-graphs.  The initial vertices in $\mathcal{I}$ are the vertices labeled with the sub-slices of $\subslicesNumber^c(\boldS_1)$
with empty in-frontier (including the empty slice), while the terminal vertices in $\mathcal{T}$ are those labeled with slices from 
$\subslicesNumber^c(\boldS_n)$ with empty out-frontier. We observe that all the dilated sub-unit-decompositions 
generated by $\subgraphsSlicegraph^c(\boldU)$ will have length $n$, irrespectively of the size of the subgraph of $G$ that 
each of them represents. Therefore each such a sub-unit decompositions 
will be potentially padded with sequences of empty slices to its left and right.
Now it should be clear that a sequence $\boldU=\boldS_1'\boldS_2'...\boldS_n'$ is a sub-unit decomposition of $\boldU=\boldS_1\boldS_2...\boldS_n$
if and only if there exists a sequence $\slicegraphvertex_{1,\boldS_1'}\slicegraphvertex_{2,\boldS_2'}...\slicegraphvertex_{n,\boldS_n'}$ labeled
with $\boldU$. Therefore the slice language $\lang(\subgraphsSlicegraph^c(\boldU))$ represents precisely the set of 
sub-unit decompositions of $\boldU$ of slice-width at most $c$.
As mentioned above, $\subgraphsSlicegraph^c(\boldU)$ is deterministic. 
This fact is guaranteed by the fact that even if a vertex in 
$\subgraphsSlicegraph^c(\boldU)$ has two forward neighbors labeled with slices carrying the same structure, their frontiers 
will forcefully have distinct numberings, and thus will be considered different.
Finally, the construction we just described can be realized in time $n\widthG^{O(c)}$ since there are at most $\binom{q}{O(c)}$ 
subslices of each slice $\boldS_i$ and we only connect vertices in $\subgraphsSlicegraph^c(\boldU)$ labeled with neighboring subslices. 
$\square$

%\subsection{\bf Proof of Theorem \ref{theorem:SliceLanguageSubgraphs}}

\begin{proposition}
\label{proposition:NumberingExtension}
Let $\slicegraph$ be a slice graph. Then a dilated unit decomposition $\boldU=\boldS_1\boldS_2...\boldS_n$ belongs to $\lang(\slicegraph)$
if and only if the unit decomposition $$\boldU'=(\boldS_1,in_1,out_1)(\boldS_2,in_2,out_2)...(\boldS_n,in_n,out_n)$$ belongs to $\lang(\numberingExtension^q(\slicegraph))$
for a set of pairs of functions $\{(in_i,out_i)\}_{i}$ where $(in_i,out_i)$ is a $q$-numbering of $\boldS_i$. 
\end{proposition}

\paragraph{\bf Proof of Theorem \ref{theorem:SliceLanguageSubgraphs}}
Let $H$ be a digraph on $k$ vertices and assume that $H$ is a subgraph of $G$ which is isomorphic to a digraph in
$\lang_{\graph}(\slicegraph)$. Since $H$ is a subgraph of $G$, there is 
a dilated unit decomposition $U'=\boldS_1'\boldS_2'...\boldS_n'$ that is a sub-unit decomposition of $\boldU$. 
By Proposition \ref{proposition:SubUnitDecompositionZigZagNumber}.$1$, $\boldU'$ has zig-zag number at most $z$, and 
therefore $\boldU'$ is compatible with some $z$-topological ordering of $\omega=(v_1,v_2,...,v_k)$ of the vertices of $H$.
Now notice that the slices $\boldS_i'$ are not normalized. Therefore there exist a normalized unit decomposition 
$\boldU''=\boldS_1''\boldS_2''...\boldS_n''$ of $H$ such that $\boldS_i'=(\boldS_i'',in_i,out_i)$ for some $q$-numbering 
$(in_i,out_i)$ of $\boldS_i''$. Since $\slicegraph$ is dilated saturated, $H\in \lang_{\graph}(\slicegraph)$ 
and $\boldU''$ also has zig-zag number $z$, we have that $\boldU''\in \lang(\slicegraph)$. Thus by Proposition \ref{proposition:NumberingExtension}, 
$\boldU'\in \lang(\numberingExtension^q(\slicegraph))$, and therefore $\boldU'\in \lang(\subgraphsSlicegraph^c(U)) \cap \lang(\numberingExtension^q(\slicegraph))$. 

Conversely, assume that the numbered dilated unit decomposition 
$$\boldU'=(\boldS_1'',in_1,out_1)(\boldS_2'',in_2,out_2) ... (\boldS_n'',in_n, out_n)$$ of the digraph $H$
belongs to the slice language $\lang(\subgraphsSlicegraph^c(\boldU)) \cap \lang(\numberingExtension^{\widthG}(\slicegraph))$.
Then by proposition \ref{proposition:NumberingExtension}, the unit decomposition $\boldU''=\boldS_1''\boldS_2''...\boldS_n''$ belongs to $\slicegraph$ 
and since $\boldU''$ is also a unit decomposition of $H$, we have that $H\in \lang_{\graph}(\slicegraph)$. Since by 
Lemma \ref{lemma:SliceGraphSubgraphs} all unit decompositions in $\lang(\subgraphsSlicegraph^c(\boldU))$ are sub-unit decompositions 
of $\boldU$ we have that $H$ is a subgraph of $G$. Therefore $H$ is a subgraph of $G$ isomorphic to a digraph in $\lang_{\graph}(\slicegraph)$. 
$\square$

\section{Proof Of Theorem \ref{theorem:MonadicTechnical}}
\label{subsection:MSO}

In the composition $\boldS_1\circ\boldS_2$ of slices defined in 
Section \ref{section:RegularSliceLanguages}, both the out-frontier vertices of $\boldS_1$ 
and the in-frontier vertices of $\boldS_2$ disappear, since they 
are not meant to be part of the structure of the composed graph. 
In this section however, it will be more convenient to consider a slightly different 
composition of slices. In this composition, which we denote by 
$\boldS_1 \oplus \boldS_2$ we simply add an edge from each out-frontier vertex 
of $\boldS_1$ to its corresponding equally numbered in-frontier 
vertex in $\boldS_2$. Indeed we will represent the existence of such 
an edge by a predicate $ConsecutiveFrontiers(u,u')$ which will be 
true whenever $u$ belongs to the out-frontier of a slice $\boldS$,
$u'$ belongs to the in-frontier of a consecutive slice $\boldS'$ and 
they have the same number. By consecutive slices we mean two slices 
that appear in consecutive positions in a slice string. 
We notice that the edge relation of graphs that arise by gluing slices according to the first composition 
can be recovered from the edge relation of the graphs that arise if they 
were composed using $\oplus$. More precisely, let 
$G^{\circ}=(V^{\circ},E^{\circ},s^{\circ},t^{\circ},l^{\circ}_V,l^{\circ}_E)=\boldS_1\circ\boldS_2\circ...\circ \boldS_n$
and $G^{\oplus}=(V^{\oplus},E^{\oplus},s^{\oplus},t^{\oplus},l^{\oplus}_V,l^{\oplus}_E) = \boldS_1\oplus\boldS_2\oplus ...\oplus
\boldS_n$. Then we define the formula $s^{\circ}(X,Y)$ to be true if and only if 
$X$ is the set of edges of a path $v_0e_1u_1u'_1e_2u_2u_2'...u_{k-1}u_{k-1}'e_kv_{\epsilon}$ in $G^{\oplus}$ such that $Y=\{v_0\}$, 
$v_{\epsilon}$ is not a frontier vertex, and $ConsecutiveFrontiers(u_i,u_i')$ holds for every $i$ with $1\leq i\leq k-1$.
Analogously, $t^{\circ}(X,Y)$ is true if $X$ is the set of edges of a path 
 $v_0e_1u_1u'_1e_2u_2u_2'...u_{k-1}u_{k-1}'e_kv_{\epsilon}$ in which $Y=\{v_{\epsilon}\}$, $v_1$ is not a frontier vertex 
and $ConsecutiveFrontiers(f_i,f_i')$ for all all intermediary vertices $u_i,u_i'$ with $1\leq i\leq k-1$. The existence of 
such a path can easily be expressed in \msotwo. Observe that 
the way in which slices are composed and the way in which slice graphs are defined will 
guarantee that for any two non-frontier vertices $v_1,v_{k+1}$ there exists at most one 
path $v_0e_1u_1u_1'e_2u_2u_2'...u_{k-1}u_{k-1}'e_kv_{\epsilon}$ such that all intermediary vertices are frontier vertices. 
Therefore each edge in $G^{\circ}$ will correspond to exactly one such a path in $G^{\oplus}$ and
vice versa. Analogously, the relations $\l^{\circ}_{V}$ and $l^{\circ}_E$ can be
easily simulated in terms of \msotwo formulas involving $l^{\oplus}_{V}$ and $l^{\oplus}_E$. 

Without loss of expressiveness, one may eliminate the need to quantify over first order variables~\cite{Thomas1997}. This will be 
in useful to reduce the number of special cases in the proof of Lemma \ref{lemma:MonadicRegularSet} below. 
The trick is to simulate first order variables via a second order predicate $singleton(X)$ 
which is interpreted as true whenever $|X|=1$ and as false otherwise. To avoid a cumbersome 
notation, whenever we refer to a variable $X$ as being a single vertex or edge, we will assume that $singleton(X)$ is true.
Since we will deal with slices, it will also be convenient to have in hands a relation $frontier(X)$ which is 
true if and only if $X$ represents a frontier vertex, and a relation $samefrontier(X,Y)$ which is true if $X$ and $Y$ 
represent vertices in the same frontier of a slice. More formally, let $\Sigma_V$ be a 
set of vertex labels, $\Sigma_E$ be a set of edge labels, $\{X,Y,..., X_1,Y_1,...\}$ be an infinite set of second order 
variables ranging over sets of vertices and let $\varphi(X)$ denote a formula with free variable $X$. 
Then the set of \msotwo formulas over directed graphs is the smallest set of formulas containing: 

\begin{itemize} 
	\item the atomic formulas, $V(X)$, $E(X)$, $singleton(X)$, $X\subseteq Y$, 
		$s(X,Y)$, $t(X,Y)$, $l_V(X,a)$ for each $a\in \Sigma_V$, $l_E(X,b)$ for each $b\in \Sigma_E$,\\
		$frontier(X)$ and $sameFrontier(X,Y)$;
	\item the formulas $\varphi \vee \psi$, $\varphi \wedge \psi$, $\neg \varphi$ and $\exists X \varphi(X)$, 
		where $\varphi$ and $\psi$ are \msotwo formulas.  
\end{itemize}

Now we follow an approach that is similar to that used in \cite{Thomas1997,Madhusudan2001} but 
lifted in such a way that it will work with slices. Let $\varphi$ be a \msotwo 
formula with $k$ free second order variables $\mathcal{X}=\{X_1,...,X_k\}$ and 
$\boldS$ be a unit slice with $r$ vertices and $r'$ edges (including the frontier vertices).  We represent an interpretation of 
$\mathcal{X}$ in $\boldS$ as a $k\times (r+r')$ boolean matrix $M$ whose rows are 
indexed by the variables in $\mathcal{X}$ and the columns are indexed by the 
vertices and edges of $\boldS$. Intuitively, we set
$M_{ij}=1$ if and only if the vertex or edge of $\boldS$ corresponding to the $j$-th column of $M$ belongs to the $i$-th
variable of $\varphi$. In this setting a sequence $M_1 M_2..M_n$ of
interpretations of a unit decomposition $\boldS_1 \boldS_2...\boldS_n$, in
which $M_i$ is an interpretation of $\boldS_i$, provides a full interpretation
of the graph $\boldS_1\oplus \boldS_2 \oplus ... \oplus \boldS_n$. Let
$\slicealphabet^c$ be the slice alphabet of width $c$. We define the interpreted 
extension of $\slicealphabet^c$ to be the alphabet 

\begin{equation*} 
\interpretedAlphabet=\bigcup_{\boldS\in \slicealphabet^c} \boldS^{\mathcal{X}} \mbox{\, where \, } \boldS^{\mathcal{X}} = \{ (\boldS,M)| M \mbox{ is an interpretation of } \mathcal{X}
\mbox{ over } \boldS\}. 
\end{equation*}

Now we are in a position to prove Lemma \ref{lemma:MonadicRegularSet}. For each formula $\varphi$
over a set of free variables $\mathcal{X}$ we will define a regular subset 
$\lang(\varphi,\interpretedAlphabet)$ of the free monoid generated by
$\interpretedAlphabet$ satisfying the following property: A string
%\vspace{-5pt}
$$(\boldS_1,M_1)(\boldS_2,M_2) ... (\boldS_n,M_n) \in (\interpretedAlphabet)^*$$
belongs to $\lang(\varphi,\interpretedAlphabet)$ if and only if the digraph 
$G=\boldS_1\oplus \boldS_2 \oplus ...\oplus \boldS_n$ satisfies $\varphi(\mathcal{X})$ with
interpretation $M_1M_2...M_n$.  
%\vspace{-5pt}

%\begin{lemma} 
%\label{lemma:MonadicRegularSet} For any \msotwo sentence $\varphi$ over digraphs and any $c\in \N$, 
%the set $\lang(\varphi,\slicealphabet^c)$ of all slice strings $\boldS_1\boldS_2...\boldS_k$ 
%over $\slicealphabet^c$ such that $\boldS_1\circ \boldS_2 \circ ...\circ \boldS_k = G$ and $G\models \varphi$ is a 
%regular subset of $(\slicealphabet^c)^*$.  
%\end{lemma} 
\paragraph{\bf Proof of Lemma \ref{lemma:MonadicRegularSet}}
\begin{proof}
%\vspace{-5pt}
By the discussion above we start by replacing each occurrence of the atomic formulas $s(X,Y)$, $t(X,Y)$, ...
in $\varphi$ by the atomic formulas $s^{\circ}(X,Y)$, $t^{\circ}(X,Y)$,... so that we can reason in 
terms of the composition $\oplus$ instead of in terms of the composition $\circ$. Let $\boldS_i=(V_i,E_i,s_i,t_i,l_{V_i},l_{E_i})$. 
First we will construct a finite automaton which accepts precisely the interpreted strings  
$$(\boldS_1,M_1)(\boldS_2,M_2)...(\boldS_n,M_n)\in (\interpretedAlphabet)^*$$ for which 
$G=\boldS_1\oplus \boldS_2 \oplus ... \oplus \boldS_n \models \varphi$ with interpretation $M_1,M_2,...,M_n$
of $\mathcal{X}$ over $G$. The proof is by induction on the structure of the formula.
It is easy to see that the atomic formulas $V(X)$, $E(X)$, $singleton(X)$, $X\subset Y$,  $l_{V_i}(X,a)$, $l_{E_i}(X)=b$ 
for each $a \in \Sigma_V$ and $b\in \Sigma_E$, can be checked by a finite automaton. For instance, 
to check whether $X_i\subseteq X_j$ holds in
$(\boldS_1,M_1)(\boldS_2,M_2)...(\boldS_n,M_n)$, the automaton verifies for each
interpretation $M_k$ with $1\leq k\leq n$ and for each column $l$ of $M_k$, that whenever $(M_k)_{il}=1$
then $(M_k)_{jl}=1$.  To determine whether $s^{\oplus}(X,Y)$ (or $t(X,Y)$) 
is true, first check whether $X,Y$ are singletons. If this is not the case, reject. Otherwise
let $X$ be interpreted as $\{e\}$ and $Y$ as $\{v\}$. Then accept either if $e$
and $v$ belong to the same slice and if $s_i(e,v)$, which can be done by table lookup. To determine
whether $ConsecutiveFrontiers(X,Y)$ is true, check whether $X$ and $Y$ are singletons, $X$, belongs
to the out-frontier of a slice and $Y$ to the in-frontier of a consecutive slice. 
Disjunction, conjunction and negation are handled by the fact that DFAs are effectively 
closed under union, intersection and complement. In other words, 
\begin{equation}
\begin{array}{c}
\lang(\varphi \vee \varphi',\interpretedAlphabet) = \lang(\varphi ,\interpretedAlphabet) \cup \lang(\varphi',\interpretedAlphabet) \\
\\
\lang(\varphi \wedge \varphi',\interpretedAlphabet) = \lang(\varphi ,\interpretedAlphabet) \cap \lang(\varphi',\interpretedAlphabet) \\
\\
\lang(\neg \varphi,\interpretedAlphabet) = \overline{\lang}(\varphi ,\interpretedAlphabet)\\
\end{array}
\end{equation}

To eliminate existential quantifiers we proceed as follows: For each variable $X$, 
define the projection $Proj_{X}:\interpretedAlphabet\rightarrow \interpretedAlphabetMinusX$  that 
sends each symbol $(\boldS,M)\in \interpretedAlphabet$ to the symbol $(\boldS,M\backslash X)$ in $\interpretedAlphabetMinusX$
where $M\backslash X$ denotes the matrix $M$ with the row corresponding to the variable $X$ deleted. 
Extend $Proj_{X}$ homomorphically to strings by applying it coordinatewise, and subsequently to languages by applying it stringwise. 
Then set
$$\lang(\exists X \varphi(\mathcal{X}), \interpretedAlphabetMinusX) = Proj_{X}(\lang(\varphi(\mathcal{X}),\interpretedAlphabet)).$$
Notice that even though homomorphisms in general do not preserve regularity, in the case of projections of symbols as defined above, this 
is not an issue. In particular one can obtain a DFA $\mathcal{A}$ accepting $\lang(\exists X \varphi(\mathcal{X},\interpretedAlphabetMinusX)$ 
from a DFA $\mathcal{A}'$ accepting $\lang(\varphi,\interpretedAlphabet)$ by simply replacing each symbol $(\boldS,M)$ appearing in a transition 
of $\mathcal{A}$ by the symbol $(\boldS,M\backslash X)$. 
At the end of this inductive process, all variables will have been projected, since $\varphi$ is a sentence. Thus the language $\lang(\varphi,\slicealphabet^c)$
will accept precisely the slice strings whose composition yield a digraph that satisfies $\varphi$. 
As a last step in our construction, we eliminate illegal sequences of slices, from the language generated by our constructed automaton, 
we intersect it with another automaton that rejects precisely the sequences of slices 
$\boldS_1\boldS_2...\boldS_n$ in which two slices that cannot be composed appear in consecutive positions.
$\square$
\end{proof} 

%%%%%%%%%%%%%%%%%%%%%%%% COMMENT THIS PART %%%%%%%%%%%%%%%
%%%%%%%%%%%%%%%%%%%%%%%% COMMENT THIS PART %%%%%%%%%%%%%%%
%%%%%%%%%%%%%%%%%%%%%%%% COMMENT THIS PART %%%%%%%%%%%%%%%
%Lemma \ref{lemma:MonadicRegularSet} gives a slice theoretic 
%analog of Courcelle's model checking theorem: In order to verify whether a digraph $G$ of existential slicewidth at most $c$ satisfy 
%a given MSO property $\varphi$, one just needs to find a slice decomposition $\boldU=\boldS_1\boldS_2...\boldS_n$ of $G$ and 
%subsequently verify whether the deterministic finite automaton accepting $\lang(\varphi,\interpretedAlphabet)$ accepts $\boldU$.
%However the goal of the present work is to make a rather different use of Lemma \ref{lemma:MonadicRegularSet}. Namely, in Theorem \ref{theorem:MonadicTechnical} 
%we will restrict Lemma \ref{lemma:MonadicRegularSet} in such a way that it generates only regular $z$-saturated slice languages, so that 
%it can be coupled to Theorem \ref{theorem:MainTechnical}.
%%%%%%%%%%%%%%%%%%%%%%%%% COMMENT THIS PART %%%%%%%%%%%%%%
%%%%%%%%%%%%%%%%%%%%%%%% COMMENT THIS PART %%%%%%%%%%%%%%%
%%%%%%%%%%%%%%%%%%%%%%%% COMMENT THIS PART %%%%%%%%%%%%%%%

For a matter of clarity, from now on we will relax our $MSO_2$ language and use lower case letters whenever referring to single edges and vertices. 
Let the predicate $PathVertices(X)$ be true whenever $X$ is the set of vertices of some path, $PathEdges(Y)$ be true whenever $Y$ is the set of edges of some path
and the  $Path(X,Y)$ be true whenever $X$ is the set of vertices and $Y$ the set of edges of the same path. 
Then the fact that a unit decomposition $\boldS_1\boldS_2...\boldS_n$ has zig-zag width at most $z$ can be expressed in \msotwo as 

\vspace{-12pt}
\begin{equation*}
\label{equation:MSOZigZagNumber}
\begin{array}{rcl}
ZigZag(z) & \equiv & (\forall X)(\forall y_1,y_2,...,y_{z+1}) \\
& & [ PathVertices(X) \wedge \bigwedge_{i} y_i \in X \Rightarrow \bigvee_{i\neq j} \neg Samefrontier(y_i,y_j)] \\
\end{array}
\end{equation*}  
\vspace{-12pt}

Basically it says that if $X$ is the set of vertices of a path and if $y_1,...,y_{z+1}$ are $z+1$ vertices in this path then at least two of 
them belong to different frontiers. We say that a digraph $G$ is $k$-path-unitable if there is a set of not necessarily edge disjoint nor vertex disjoint 
paths $\{p_1,...,p_k\}$ such that $G=(V,E)=\cup_{i=1}^k p_i$. The fact that a graph $G$ is $k$-path-unitable can be expressed by 
the formula

\vspace{-6pt}
\begin{equation*}
\label{equation:MSOkPaths}
\begin{array}{rcl}
Unitable(k) & \equiv & (\exists X_1,...,X_k, Y_1,...,Y_k) \\
&  & \left[ V\subseteq \bigcup_i X_i \wedge E\subseteq \bigcup_i Y_i \wedge   \bigwedge_{i} Path(X_i,Y_i)\right] \\
\end{array}
\end{equation*}

\vspace{-6pt}
\paragraph{{\bf Proof of Theorem \ref{theorem:MonadicTechnical}}}
Let $\varphi'=\varphi \wedge ZigZag(z) \wedge Unitable(k)$. By Lemma \ref{lemma:MonadicRegularSet}, there is a 
regular slice language $\lang^{\varphi,k,z}$  over $\slicealphabet^{k\cdot z}$ generating all unit decompositions over $\slicealphabet^{k\cdot z}$
whose composition yields a graph $G$ satisfying $\varphi'$, and in particular $\varphi$. Since the factor $Unitable(k)$ is present in $\varphi'$
all these graphs can be cast as the union of $k$-paths. Since the factor $ZigZag(z)$ is present in $\varphi'$, all unit 
decompositions in $\lang^{\varphi,k,z}$ have zig-zag number at most $z$. It remains to show that every unit decomposition 
of zig-zag number at most $z$ of a graph $H\in \lang_{\graph}^{\varphi,k,z}$ is in $\lang^{\varphi,k,z}$. This follows from the fact that $G$ is 
the union of $k$ directed paths, and from Proposition \ref{proposition:SliceWidthFacts}.\ref{item:ZigZagNumberPathCovering} 
stating that any unit decomposition of zig-zag number at most $z$ of a digraph that is the union of at most $k$ directed paths
has slice-width at most $k\cdot z$. Thus $\lang^{\varphi,k,z}$ is $z$-saturated. To finish the proof set 
$\slicegraph(\varphi,k,z)$ as any slice graph generating  $\lang^{\varphi,k,z}$ $\square$

\end{document}